\documentclass[moor,sglanonrev]{informs4}
\usepackage{eqndefns-left} % For checking the display equation width and equation environment definitions %
\usepackage[utf8]{inputenc}

\RequirePackage{tgtermes}
\RequirePackage{newtxtext}
\RequirePackage{newtxmath}
\RequirePackage{bm}
\RequirePackage{endnotes}

\OneAndAHalfSpacedXII % Current default line spacing
\usepackage[sort&compress]{natbib}
 \bibpunct[, ]{[}{]}{,}{n}{}{,}%
 \def\bibfont{\small}%

\usepackage{amsmath}

\usepackage{amssymb}

\usepackage{graphicx}
\usepackage{enumitem}
\usepackage{cleveref}
\usepackage{nameref}
\usepackage{tikz}%tkz-graph,tkz-berge}
\usetikzlibrary{decorations.pathreplacing}
\usetikzlibrary{positioning,chains,fit,shapes,calc,arrows}
\usepackage{amsfonts}
\usepackage{verbatim}
\usepackage{mathtools}
\usepackage{multirow}
\usepackage{array}
\usepackage[font=small]{subcaption}
\usepackage[linesnumbered,ruled,vlined]{algorithm2e}
\usepackage{xfrac}
\usepackage{breqn}
\usepackage{thm-restate}
\usepackage{algorithmicx}
\usepackage{color}

\newlength{\proofpostskipamount}\newlength{\proofpreskipamount}
\renewcommand{\contentsname}{Contents\\ }

\renewenvironment{proof}%
               {\par\vspace{\proofpreskipamount}\noindent{\bf Proof:}\hspace{0.5em}}% 0.5 before
               {\nopagebreak%
                \strut\nopagebreak%
                \hspace{\fill}\qed\par\vspace{\proofpostskipamount}\noindent}

\newcommand{\set}[2]{\{#1\, ; \, #2\}}

\newcommand{\modulus}[1]{\vert #1 \vert}
\newcommand{\goods}{M}
\newcommand{\agents}{N}
\newcommand{\half}{\sfrac{1}{2}}
\newcommand{\OH}{O^H}
\newcommand{\nhalf}{n_{\sfrac{1}{2}}}
\newcommand{\mhalf}{m_{\sfrac{1}{2}}}
\newcommand{\Ahalf}{\SA_{\sfrac{1}{2}}}
\newcommand{\Ohalf}{\SO_{\sfrac{1}{2}}}

\newcommand{\htmladdnormallink}[1]{#1}

\newcommand{\NSW}{\text{NSW}}
\newcommand{\OPT}{\text{OPT}}
\newcommand{\instance}{\mathcal{I}}

\newcommand{\NN}{\ensuremath{\mathbb{N}}}
\newcommand{\RR}{\ensuremath{\mathbb{R}}}
\newcommand{\QQ}{\ensuremath{\mathbb{Q}}}

\newcommand{\classNP}{\textsf{NP}}
\newcommand{\classAPX}{\textsf{APX}}

\newcommand{\OPTH}{\OPT^H}
\newcommand{\EH}{E^H}
\newcommand{\EL}{E^L}
\newcommand{\AH}{A^H}

\newcommand{\CH}{C^H}

\newcommand{\uv}{v}
\newcommand{\abs}[1]{ \lvert #1 \rvert }
\newcommand{\ceil}[1]{\lceil #1 \rceil}
\newcommand{\floor}[1]{\lfloor #1 \rfloor}

\newtheorem{observation}{Observation}
\newcommand{\hdeg}{\mathit{hdeg}}
\newcommand{\sset}[1]{\{ #1 \}}

\newcommand{\assign}{\mathbin{\raisebox{0.05ex}{\mbox{\rm :}}\!\!=}}

\newcommand{\myhide}[1]{}

\EquationsNumberedThrough    % Default: (1), (2), ...
\TheoremsNumberedThrough     % Preferred (Theorem 1, Lemma 1, Theorem 2)
\ECRepeatTheorems  %  

\MANUSCRIPTNO{MOOR-0001-2024.00}

\begin{document}

\ARTICLEAUTHORS{
\AUTHOR{Hannaneh Akrami}
\AFF{Max Planck Institute for Informatics, Saarbr\"ucken, and Hertz Chair for Algorithms and Optimization, Bonn University, \EMAIL{hannaneh.akrami95@gmail.com}}
\AUTHOR{Bhaskar Ray Chaudhury}
\AFF{University of Illinois, Urbana-Champaign, Department of Computer Science and Department of Industrial and Enterprise Systems Engineering, \EMAIL{braycha@illinois.edu}}
\AUTHOR{Martin Hoefer}
\AFF{RWTH Aachen University, Department of Computer Science, \EMAIL{mhoefer@cs.rwth-aachen.de}}
\AUTHOR{Kurt Mehlhorn}
\AFF{Max Planck Institute for Informatics and Universit\"at des Saarlandes, \EMAIL{mehlhorn@mpi-inf.mpg.de}}
\AUTHOR{Marco Schmalhofer}
\AFF{Goethe University Frankfurt, Institute for Computer Science, \EMAIL{schmalhofer@em.uni-frankfurt.de}}
\AUTHOR{Golnoosh Shahkarami}
\AFF{Max Planck Institute for Informatics and Saarbr\"ucken Graduate School of Computer Science, Universit\"at des Saarlandes, \EMAIL{gshahkar@mpi-inf.mpg.de}}
\AUTHOR{Giovanna Varricchio}
\AFF{University of Calabria, Department of Mathematics and Computer Science, \EMAIL{giovanna.varricchio@unical.it}}
\AUTHOR{Quentin Vermande}
\AFF{Inria Centre at Université C\^{o}te d'Azur, \EMAIL{qvermande@phare.normalesup.org}}
\AUTHOR{Ernest van Wijland}
\AFF{IRIF, Université Paris-Cit\'{e}, \EMAIL{ernest.vanwijland@irif.fr}}
}

\RUNAUTHOR{Akrami et al.}
\RUNTITLE{Maximizing Nash Social Welfare in 2-Value Instances}
\TITLE{Maximizing Nash Social Welfare in 2-Value Instances:\\
Delineating Tractability}

\ABSTRACT{
We study the problem of allocating a set of indivisible goods among a set of agents with \emph{2-value additive valuations}. In this setting, each good is valued either $1$ or $\sfrac{p}{q}$ for some fixed co-prime numbers $p,q\in \NN$ such that $1\leq q < p$. Our goal is to find an allocation that maximizes the \emph{Nash social welfare} (\NSW), i.e.,  the geometric mean of the valuations of the agents. In this work, we give a complete characterization of polynomial-time tractability of \NSW\ maximization that solely depends on the value of $q$.

We start by providing a rather simple polynomial-time algorithm to find a maximum \NSW\ allocation when the valuation functions are \emph{integral}, that is, $q=1$. We then exploit more involved techniques to get an algorithm that produces a maximum \NSW\ allocation for the \emph{half-integral} case, that is, $q=2$.  Finally, we show it is \classNP-hard to compute an allocation with maximum \NSW\ whenever $q\geq3$.
}

\FUNDING{Martin Hoefer and Giovanna Varricchio were supported by DFG grant Ho 3831/5-1.}


\tableofcontents

\maketitle

%\tableofcontents

\section{Introduction}

Fair division of goods has developed into a fundamental field in economics and computer science. In a classical fair division problem, the goal is to allocate a set of goods among a set of agents in a \emph{fair} (making every agent satisfied with her bundle) and \emph{efficient} (achieving good overall welfare) manner. One of the most well-studied classes of valuation functions is the one of \emph{additive} valuation functions, where the utility of a bundle is defined as the sum of the utilities of the contained goods. When agents have additive valuation functions, Nash social welfare (NSW), or equivalently, the geometric mean of the valuations, is a direct indicator of the fairness and efficiency of an allocation. In particular, Caragiannis et al.~\cite{CaragiannisKMPSW19} show that any allocation that maximizes NSW is \emph{envy-free up to one good} (EF1), i.e.,  no agent envies another agent after the removal of \emph{some} single good from the other agent's bundle, and \emph{Pareto-optimal}, i.e., no allocation gives a single agent a better bundle without giving a worse bundle to some other agent.  Maximizing NSW is APX-hard. Also, allocations achieving good approximations of NSW may not have similar fairness and efficiency guarantees~\cite{Lee17}. Despite this, approximation algorithms are of great interest, and algorithms with small constant approximation factors have been obtained recently~\cite{AnariGSS17, BarmanKV18, ColeDGJMVY17, ColeG18}. The current best factor is $e^{1/e} \approx 1.445$, Barman et al.~\cite{BarmanKV18}. Their algorithm uses techniques inspired by competitive equilibria along with suitable rounding of valuations to guarantee polynomial running time.

While computing an allocation with maximum NSW is generally hard, it becomes computationally tractable when the agents have \emph{binary additive valuations}, i.e., when for each agent $i$ and each good $g$, we have $v_i(g) \in \sset{0,1}$, Barman et al.~\cite{BarmanKV18AAMAS}. Although this class of valuation functions may seem restrictive in its expressiveness, several real-world scenarios involve binary preferences, and, in fact, there is substantial research on fair division under binary valuations~\cite{AleksandrovAGW15, BarmanKV18AAMAS, BouveretL16, DarmannS15, FreemanSVX19, HalpernPPS20, suksompong2022maximum}. Furthermore, even for asymmetric agents, i.e., agents with different entitlements, a maximum NSW allocation satisfies strategyproofness, together with other interesting properties~\cite{HalpernPPS20,suksompong2022maximum}. 

A generalization of binary valuation functions are \emph{2-value} functions, where for each agent $i$ and each good $g$, we have $v_i(g) \in \sset{r,s}$ for some $r,s \in \QQ$. The case $r = s$ is trivial as every agent gives the same value to all the goods. Binary valuations are the special case  $r=0$ and $s=1$. Amanatidis et al.~\cite{AmanatidisBFHV21} show that for 2-value functions, an allocation with maximum NSW is \emph{envy-free up to any good} (EFX), where no agent envies another agent following the removal of \emph{any} single good from the other agent's bundle. The authors also present a polynomial-time algorithm that provides an EFX allocation for 2-value instances. However, the computed allocation is neither guaranteed to maximize NSW, nor to be Pareto optimal. Garg and Murhekar~\cite{GargM21} show how to obtain an EFX and PO allocation efficiently and also provide a $1.061$-approximation algorithm for the maximum \NSW\ in 2-value instances; this was improved to 1.0345 in the conference version of this paper~\cite{Akrami_Chaudhury_Hoefer_Mehlhorn_Schmalhofer_Shahkarami_Varricchio_Vermande_Wijland_2022}. 
Both Amanatidis et al.~\cite{AmanatidisBFHV21} and Garg and Murhekar~\cite{GargM21} left the problem of (exactly) maximizing NSW for 2-value instances open. For 3-value instances, maximizing NSW is NP-complete~\cite{AmanatidisBFHV21}. 

%Finding an allocation with maximum Nash social welfare for 2-value instances would provide a fair (EFX) and efficient (PO) allocation, Amanatidis et al.~\cite{AmanatidisBFHV21}; however, both Amanatidis et al.~\cite{AmanatidisBFHV21} and Garg et al.~\cite{GargM21}, left the problem of maximizing Nash social welfare for 2-value instances open. 

\subsection{Our Contribution}
In this paper, we characterize the complexity of maximizing NSW for 2-value instances. Surprisingly, the tractability of this problem changes according to the ratio between the two values $r$ and $s$. 
%More precisely, if we set $a=\sfrac{p}{q}$ and $b=1$\footnote{While considering maximum Nash welfare, the normalization of the agents valuations over goods does not affect the space of the optimal solutions.}, where $p,q \in \NN$ and $1\leq p < q$, then there exists a polynomial-time algorithm for maximizing Nash social welfare when $p=1,2$ and it is \classNP-hard to determine a maximum Nash welfare allocation if $p\geq 3$. Furthermore, we show that the problem is  \classAPX-hard for $p=4$ and $q=5$. Our contribution is summarized by the following theorems.
Since scaling an agent's valuation by a uniform factor for all goods does not affect the optimality properties of allocations, let us assume w.l.o.g.\ that $r=1$ and $s=\sfrac{p}{q}$, for some coprime numbers $p,q \in \NN$ such that $1\leq q < p$. 
%Our main contribution is summarized by the following theorems.

We first show two positive results. If $q$ is either $1$ or $2$, there is a polynomial-time algorithm that computes a maximum \NSW\ allocation.

\begin{restatable}{theorem}{IntegerResult}\label{IntegerResult}
There exists a polynomial-time algorithm that computes  a maximum \NSW\ allocation for integral instances, i.e., when $q = 1$ and $p$ is an integer greater than one. 
\end{restatable}

\begin{restatable}{theorem}{HalfIntegerResult}\label{HalfIntegerResult}
There exists a polynomial-time algorithm that computes  a maximum \NSW\ allocation for half-integral instances, i.e., when $q = 2$ and $p$ is an odd integer greater than two. 
\end{restatable}

The proof and the algorithm for the half-integer result are considerably more complex than the ones for the integer case. We complete the characterization in terms of computational complexity by showing \classNP-hardness in the remaining cases, i.e., when $q\geq 3$.

\begin{restatable}{theorem}{NPhard}\label{NPhard}
It is \classNP-hard to compute an allocation with optimal NSW for 2-value instances, for any constant coprime integers $p>q\geq 3$. 
\end{restatable}

\iffalse
\begin{theorem}
There exists a polynomial-time algorithm computing a maximum \NSW\ allocation for integral instances, i.e., when $q = 1$ and $p$ is an integer greater than one. 
\end{theorem}

\begin{theorem}
There exists a polynomial-time algorithm computing a maximum \NSW\ allocation for half-integral instances, i.e., when $q = 2$ and $p$ is an odd integer greater than two. 
\end{theorem}

We complete the characterization in terms of computational complexity by showing \classNP-hardness in the remaining cases, i.e., when $q\geq 3$.

\begin{theorem}
It is \classNP-hard to compute an allocation with optimal \NSW\ for 2-value instances, for any constant coprime integers $p>q\geq 3$.
\end{theorem}
\fi

%The analysis of our algorithms requires a number of novel technical contributions. In the next section, we highlight the main ideas and techniques.

We prove Theorem 1 in Section~\ref{sec:integer}, Theorem 2 in Section~\ref{sec:halfInteger}, and Theorem 3 in Section~\ref{sec:NPhard}. 
Throughout the paper, we write $s$ for $\sfrac{p}{q}$.

\subsection{Our Techniques}
\label{sec:techniques}

We interpret the problem of maximizing \NSW\ as a graph problem: We have an edge-weighted complete bipartite graph with the set of agents and the set of goods being the two sides. The weight of an edge represents the value of a good for an agent. It is either $1$ (in which case we call the edge a \emph{light edge}) or $s$ (called a \emph{heavy edge}). We say that a good is heavy if it has at least one incident heavy edge and is light otherwise. A (partial) allocation is a multi-matching in which every good has a degree of at most one. We call an allocation complete if all goods have degree one. This representation of allocations allows us to use the concept of alternating paths and alternating walks. An alternating walk is a sequence of alternating paths glued together at their endpoints. We next outline our algorithms, see Algorithm~\ref{The Algorithm}. The algorithm for the half-integral case has five steps. Only steps 1, 2, and 4 are needed for the integral case, and only step 1 is needed for the binary case, as the allocation of the goods of value 0 can be arbitrary. In steps 1 to 3, a heavy good can only be allocated to an agent that considers it heavy.

\begin{algorithm}[t]\smallskip
  
\begin{enumerate}
\item Determine the lexmin allocation of the heavy goods, i.e., push heavy goods towards smaller bundles as much as possible. Let $b_n \ge b_{n-1} \ge \ldots \ge b_1$ be the number of heavy goods allocated to the different agents sorted in decreasing order. Then, $b_n$ is minimal among all allocations of the heavy goods, and given that $b_n$ has its minimal value, $b_{n-1}$ is minimal, and so on. In other words, the string $b_nb_{n-1}\ldots b_1$ is lexicographically minimal.
\item Allocate the light goods greedily, i.e., allocate the light goods one by one, and always add the next good to a bundle of smallest value.
\item Let $x$ be the minimum value of any bundle in the resulting allocation. Call the bundles of value $x$, $x + \half$, and $x + 1$ the small bundles, and let $N_s$ be the owners of the small bundles. Optimize the allocation of the small bundles, i.e., allocate the goods contained in the small bundles to the agents in $N_s$ so that each bundle has a value in $\sset{x, x + \sfrac{1}{2}, x + 1}$ and the number of bundles of value $x + \half$ is maximum.
\item Take a heavy item from a bundle of maximum value and move it (as a light item) to a bundle of minimum value.
\item Reoptimize the allocation of the small bundles as in step 3. %If step 4 (in the integer case) or steps 4 and 5 (in the half-integral) case improved the allocation, repeat. If not reverse steps 4 or  4 and 5.
\end{enumerate}
\caption[t]{\label{The Algorithm}  Algorithm for the half-integral case. The output of the algorithm is the best allocation (= maximum NSW) obtained in steps 4 and 5. Only steps 1, 2, and 4 are needed for the integral case, and only step 1 is needed for the binary case. In steps 1 to 3, a heavy good can only be allocated to an agent that considers it heavy.   }
\end{algorithm}

\begin{example} \rm
\label{ex:int-halfInt}
Suppose there are two agents with identical valuations. We have five goods, of which two are heavy. It is easy to see that any optimal solution follows one of two patterns: (1) assign both heavy goods to one agent, all light goods to the other agent (NSW =  $6s$); (2) assign one heavy good to each agent, and, in addition, one agent gets two light goods, the other gets one light good (NSW = $(s + 1)(s + 2)$). Whether optimal allocations follow the first and/or the second pattern depends on the value of $s$. In particular, all optimal allocations follow the first pattern if and only if $s < 2$; all follow the second one if and only if $s > 2$; both patterns are optimal when $s = 2$. Hence, depending on the ratio $s$, the distribution of heavy and light goods in optimal allocations may change. In particular, the first pattern yields an unbalanced allocation of heavy goods -- one agent receives both heavy goods while the other none. In contrast, in the second pattern, the heavy goods are balanced and their allocation is lexmin. In the case $s = \sfrac{3}{2}$, step 3 is needed to optimize the allocation of the small bundles.  \mbox{}\hfill $\blacksquare$
\end{example}

%These observations motivate the main structure of our algorithms: depending on the ratio between heavy and light goods values, we allocate the heavy goods carefully and then allocate the light goods accordingly. 
It turns out that understanding how heavy goods are distributed in optimal allocations is the key challenge in computing a maximum \NSW\ allocation. We characterize the allocation of heavy goods and use these insights to design efficient algorithms.

\paragraph{Characterizing the allocation of heavy goods:}
%The main challenge in computing an allocation with maximum \NSW\ is to identify the distribution of heavy goods. 
First, we consider instances with $q=1$ (which we call \emph{integral} instances). We give a concise characterization of the distribution of heavy goods in a maximum \NSW\ allocation. We refer to the \emph{heavy-part} $A^H$ of an allocation $A$ as the set of all heavy edges in the allocation, and we call an allocation \emph{heavy-only} if the allocation contains only heavy edges, i.e., if $A^H = A$. One of our main structural results (shown in Lemma \ref{heavysim}) is that there exists a maximum \NSW\ allocation $\OPT$ such that the heavy-part of $\OPT$ is lexicographically minimum (\emph{lexmin}) among all heavy-only allocations of the same cardinality. Therefore, if we know the number of heavy-edges in $\OPT$, then the utility profile of the heavy-part of $\OPT$ is unique (as it is lexmin).

For instances with $q=2$ (called \emph{half-integral} instances), the lexmin property is not necessarily satisfied. The lexmin property may be interpreted as ``the distribution of heavy goods is balanced as much as possible''. Here, however, the heavy-part of the allocation might have to be unbalanced in order to maximize the \NSW\ (see Example \ref{ex:int-halfInt}). Let $x$ be the minimum value of any bundle after step 2. Then only bundles of value $x$, $x + \half$ and $x + 1$ can contain light items. Call these the small bundles. The alloation of the non-small bundles is optimal after step 2. The goal of step 3 is to reallocate the items in the small bundles so as to maximize the number of bundles of value $x + \half$. Note that the number of heavy items in bundles of value $x$ and $x + 1$ have the same parity, which is different from the parity of the number of heavy items in bundles of value $x + \half$. So, the number of heavy items in a bundle of value $x + d$, $d \in \sset{0,\half,1}$ lies in $S_d \assign \sset{\ell_d,\ell_d + 2,\ldots,r_d}$, where $r_d = \floor{(x + d)/s}$, $\ell_d \in \sset{0,1}$, $\ell_0 = \ell_1$, and $\ell_0 \not= \ell_{\half}$. The question of whether we can create additional bundles of value $x + \half$ then amounts to reallocating the heavy items 
so that more agents have a number of heavy items in $S_\half$, and the number of heavy items allocated to any agent does not exceed the bounds set by the $r_d$-values. We will show that this question can be answered by means of a polynomially solvable matching problem~\cite{Lovasz70,CornuejolsFactors}, see Section~\ref{sec:halfInteger} for details.

\paragraph{Allowing heavy goods to be allocated as light goods:} This step is conceptually quite simple, although the correctness proof is surprisingly involved. We take a heavy good from a most valuable bundle and move it to a least valuable bundle. In the half-integral case, we then reoptimize the allocation of the small bundles. We repeat and choose the best allocation obtained.

\paragraph{Hardness result:}
For instances with $q\geq 3$, it is \classNP-hard to compute a maximum \NSW\ allocation. Our proof is based on a reduction from Exact $q$-Dimensional Matching. Our proof formalizes the intuition that \emph{determining the structure of heavy goods in an optimal allocation} is the main challenge in maximizing the \NSW. In particular, we show that it is \classNP-hard to determine the distribution of heavy goods in an optimal allocation.

\subsection{History of the Paper and Further Related Work}

The present paper is based on three earlier papers~\cite{Akrami_Chaudhury_Hoefer_Mehlhorn_Schmalhofer_Shahkarami_Varricchio_Vermande_Wijland_2022, NSW-twovalues-halfinteger,NSW-twovalues-simplified}. The first paper presented our algorithm for integral instances and the \classNP-hardness results in the current paper as part of an extended abstract. We also discussed \classAPX-hardness for $q\ge 4$ and showed that our algorithm for integral instances has an approximation factor of at most $\frac{24}{29} \exp \left(\frac{110}{493}\right) < 1.0345$ for general 2-value instances. For the present paper, we decided to focus on the computational complexity of optimal solutions and omit the consideration of approximations. The second paper filled in the missing details and provided a solution for the half-integer case. The third paper considerably simplified the half-integer case. The second paper was accepted to MOR in November of '24, at around the time when we found the simplified proof. We submitted the simplified proof to MOR, and the editors suggested to merge the two papers. We followed the suggestion. 

Beyond additive valuations, the design of approximation algorithms for \emph{submodular valuations} received considerable attention, see Anari et al.~\cite{AnariMGV18}, Garg et al.~\cite{GargHM18}, Chaudhury et al.~\cite{ChaudhuryCGGHM18}), Garg et al.~\cite{GargHV21}, and Li and Vondrak~\cite{LiV21}. The currently best approximation ratio for submodular valuations is $4+\epsilon$ by Garg et al.~\cite{garg2022approximating}.

For \emph{binary} submodular valuations where the marginal value of every agent for every good is either $0$ or $1$, an allocation maximizing the \NSW\ can be computed in polynomial time, Babaioff et al.~\cite{BabaioffEF21}. 
In particular, in this case, one can in polynomial time find an allocation that is Lorenz dominating, simultaneously minimizes the lexicographic vector of valuations, and maximizes both utilitarian social welfare, i.e., the sum of the agents' utilities, and NSW. Moreover, this allocation is also strategyproof. 

More generally, there are approximation algorithms for maximizing \NSW\ with subadditive valuations, Barman et al.~\cite{BarmanBKS20}, Chaudhury et al.~\cite{ChaudhuryGM21}, Dobzinski et al.~\cite{Dobzinski}, and even asymmetric agents, Garg et al.~\cite{GargKK20} and Brown et al.~\cite{Brown-Laddha}.  

There is also literature on guaranteeing high \NSW\ with other fairness notions. For instance, relaxations of EFX can be guaranteed with high \NSW, Caragiannis et al.~\cite{CaragiannisGravin19} and Chaudhury et al.~\cite{ChaudhuryGM21}. Moreover, approximations of \emph{groupwise maximin share} (GMMS), Chaudhury et al.~\cite{CKMS20}, and \emph{maximin share} (MMS), Caragiannis et al.~\cite{CaragiannisKMPSW19} and Chaudhury et al.~\cite{CKMS20}, are achieved with high \NSW.

\subsection{Organization}
The rest of this paper is structured as follows. We start by providing definitions and notations in Section~\ref{sec:preliminaries}. In Section~\ref{sec:integer}, we discuss the polynomial-time algorithm for integral instances. In Section~\ref{sec:halfInteger}, we treat half-integral instances. Finally, in Section~\ref{sec:NPhard}, we show that for other classes of 2-value instances the problem of maximizing the \NSW\ is \classNP-hard.

\section{Preliminaries}\label{sec:preliminaries}

    A fair division instance $\instance$ is given by a triple $(\agents, \goods, (\uv_1, \dots, \uv_n))$, where $\agents$ is a set of $n \ge 1$ agents,  $\goods$ is a set of $m\geq n$ indivisible goods, and $\uv_i:2^\goods \rightarrow \RR_{\geq 0} $ is the \emph{valuation function} of agent $i$. Valuation functions are additive, i.e., $\uv_i(X) = \sum_{g \in X} \uv_{i}(\sset{g})$ for every $X \subseteq \goods$. For the sake of simplicity, we write $\uv_{i}(g)$ instead of $\uv_{i}(\sset{g})$. 
    
    In this paper, we study \emph{2-value} additive valuations, in which, for each $g\in\goods$, $\uv_i(g) \in \sset{1,s}$ for $s = \sfrac{p}{q}$ and fixed $p,q \in \NN$. To avoid trivialities, we assume $1 \le q< p$, and $p,q$ to be coprime numbers.
An \emph{allocation} $A=(A_1, \dots, A_n)$ is a partition of $\goods$ among the agents, where $A_i \cap A_j = \emptyset$, for each $i \neq j$, and $\bigcup_{i\in \agents} A_i = \goods$. We evaluate an allocation using the Nash social welfare 
    %$\NSW(A) = \left(\prod _{i\in N}\uv_i(A_i)\right)^{\sfrac{1}{n}}$. 
    \[
    \NSW(A) = \left(\ \prod _{i\in N}\uv_i(A_i)\ \right)^{\sfrac{1}{n}} \enspace .
    \]
    When comparing the $\NSW$ of two allocations, we frequently drop the power of $1/n$ for simplicity. 

    \subsection{Utility Graphs and Utility Profiles}
    In our paper, we exploit the connection between the concepts of allocation in fair division and of one-to-many matchings in bipartite graphs. For this reason, we find it more convenient to define an allocation in the context of a bipartite graph. Consider the complete bipartite graph $G = (N \cup M, E)$, where we have agents on one side and goods on the other side. We call the edge between agent $i$ and good $g$ \emph{heavy} if $\uv_i(g) =  s$ and \emph{light} otherwise. We use $\EH$ and $\EL$ to denote the set of heavy and light edges, respectively.  Moreover, good $g$ is \emph{heavy} for agent $i$ if $\uv_i(g) = s$ and \emph{light} otherwise. %Figure~\ref{fig1}a shows an instance with two agents and three goods.
	
	An \emph{allocation} is a subset $A\subseteq E$ such that for each $g \in M$ there is at most one edge in $A$ incident to $g$. Note that according to this definition allocations may be partial.
%	Nonetheless, in this work we aim to compute complete allocations, that is, for each $g\in\goods$ there exists an edge $e\in A$ that is incident to $g$.
	If there is an edge $(i,g) \in A$, we say that $g$ is \emph{assigned} to $i$ in $A$, or $i$ \emph{owns} $g$ in $A$, or $A$ \emph{assigns} $g$ to $i$. Otherwise, $g$ is \emph{unassigned}. An allocation is \emph{complete} if all goods are assigned. For an agent $i$, we use $A_i$ to denote the set of goods assigned to $i$ in $A$. We call $A_i$ the \emph{bundle} of $i$ in $A$. Then, $\uv_i(A_i)$ is the utility of $i$'s bundle for $i$. %Figure~\ref{fig2}b shows an allocation for the instance depicted in Figure~\ref{fig1}a.
	
        % 	\begin{figure}[tb!]
    %     \centering
    %     \begin{subfigure}{7cm}
    %         \centering
    %         \begin{tikzpicture}[every node/.style={circle,draw,font=\sffamily\small\bfseries, inner sep = 3pt}]
            
    %         \Vertices
            
    %         % the edges
    %         \foreach \i in {1,2}
    %             \foreach \j in {1,2,3}
    %                 \draw[-, gray] (a\i) -- (g\j);
                    
    %         \draw[-, very thick] (a1) -- (g1);
    %         \draw[-, very thick] (a1) -- (g2);
    %         \draw[-, very thick] (a2) -- (g2);
    %         \draw[-, very thick] (a2) -- (g3);
           
    %     \end{tikzpicture}
        
    %     (a) Graph representation.
    %     \end{subfigure}
    %     \begin{subfigure}{7cm}
    %         \centering
    %         \begin{tikzpicture}[every node/.style={circle,draw,font=\sffamily\small\bfseries, inner sep = 3pt}]
            
    %         \Vertices
            
    %         % the edges
    %         \draw[-, very thick] (a1) -- (g1);
    %         \draw[-, very thick] (a2) -- (g2);
    %         \draw[-, gray] (a1) -- (g3);
           
    %     \end{tikzpicture}
        
    %         (b) An allocation.
    %        \end{subfigure}
    %        \caption{The depicted graph corresponds to an instance with two agents and three goods. Thick black edges and thin gray edges correspond to heavy and light edges, respectively.\label{fig1}\label{fig2}
    %      %  The graph $G$ in (b) corresponds to a complete allocation $A$ in which $g_1$ and $g_3$ are allocated to $a_1$ and $g_2$ is allocated to $a_2$. Note that this allocation does not maximize $\NSW$. $G$ restricted to black edges is $\AH$. We have $deg_H(a_1)=deg_H(a_2)=1$.
    %      }
    % \end{figure}

	The \emph{utility vector} of an allocation $A$ is given by $(\uv_1(A_1 ),$ $\ldots, \uv_n(A_n))$, and its \emph{utility profile} is obtained by rearranging its components in non-decreasing order. A utility profile $(a_1,\ldots,a_n)$ is \emph{lexicographically smaller}\footnote{Note that $a_1 \le a_2 \le \ldots \le a_n$ in a utility profile. Then $(a_1,\ldots,a_n) <_{\mathit{lex}} (b_1,\ldots b_n)$ if the string $a_n\ldots a_1$ is lexicographically smaller than the string $b_n\ldots b_1$.} than a utility profile $(b_1,\ldots,b_n)$ (denoted by $(a_1,\ldots,a_n)<_{\mathit{lex}}(b_1,\ldots,b_n)$) if the profiles are different and $a_i < b_i$ for the largest $i$ with $a_i \not= b_i$. An allocation $A$ with utility profile $(a_1,\ldots,a_n)$ is \emph{lexmin} in a family $\cal A$ of allocations if there is no allocation $B \in \cal A$ with utility profile $(b_1,\ldots,b_n)$ such that $(b_1,\ldots,b_n)<_{\mathit{lex}} (a_1,\ldots,a_n)$.
	
    For an allocation $A$, its \emph{heavy part} $\AH$ is the restriction of $A$ to the heavy edges, i.e., $\AH = A \cap \EH$. An allocation $A$ is \emph{heavy-only} if $A = \AH$.
	For an agent $i$, $\AH_i$ is the set of heavy edges incident to agent $i$ under allocation $A$.
	We refer to $\modulus{\AH_i}$ as the \emph{heavy degree} of $i$ in $A$ and denote it by $\deg_H(i,A)$ or $\deg_H(i)$.

	\subsection{Alternating Paths}\label{altpath}
	We reformulated the fair division setting so that allocations correspond to multi-matchings. This is motivated by the fact that, in our algorithms, we improve the \NSW\ of an allocation using the notion of {\em alternating paths}. %We also use other techniques. 
	
	An \emph{alternating path} with respect to an allocation $A$ is any path whose edges alternate between $A$ and $E \setminus A$.
    Alternating paths with only heavy edges will be of particular interest. A \emph{heavy alternating path} is an alternating path  whose edges belong to $\EH$. 
	%In Figure~\ref{fig:heavyAlternating}a, we give an example of a heavy alternating path for the instance depicted in Figure~\ref{fig1}a.
	
   %  \begin{figure}[tb!]
   %      \centering
   %      \begin{subfigure}{9cm}
   %          \centering
   %          \begin{tikzpicture}[every node/.style={circle,draw,font=\sffamily\small\bfseries, inner sep = 3pt}]
            
   %          \Vertices
            
   %          % the edges
   %          \draw[-, very thick] (a1) -- (g1);
   %          \draw[-, very thick] (a2) -- (g2);
   %          \draw[-, dashed, very thick] (a1) -- (g2);
   %          \draw[-, dashed, very thick] (a2) -- (g3);
            
   %          \end{tikzpicture}
            
   %      (a) Allocation $A$ and the heavy alternating path $P$.
           
   %      \end{subfigure}
   %      \begin{subfigure}{7cm}
   %          \centering
   %          \begin{tikzpicture}[every node/.style={circle,draw,font=\sffamily\small\bfseries, inner sep = 3pt}]
            
   %          \Vertices
            
   %          % the edges
   %          \draw[-, dashed, very thick] (a1) -- (g1);
   %          \draw[-, dashed, very thick] (a2) -- (g2);
   %          \draw[-, very thick] (a1) -- (g2);
   %          \draw[-, very thick] (a2) -- (g3);
            
   %      \end{tikzpicture}

   %          (b) $A \oplus P$.
   %      \end{subfigure}
   %      \caption{
   % %     The path $P=(g_1, a_1, g_2, a_2, g_3)$ in (a) is a heavy alternating path where all the edges are heavy and the thick black edges shows allocation $A$. (b) shows $A \oplus P$.
   % In both figures the path $P=(g_1, 1, g_2, 2, g_3)$ is depicted. Solid edges represent the allocation.}
   %      \label{fig3} \label{fig:heavyAlternating} \label{fig:A+P}
   %  \end{figure}

	    An \emph{alternating path} with respect to two allocations $A$ and $B$ is any path whose edges alternate between $A \setminus B$ and $B \setminus A$, i.e., between edges only in $A$ and edges only in $B$. %but not in both $A$ and $B$. 
	
	An \emph{alternating path decomposition} is defined with respect to two heavy-only allocations $A$ and $B$. The graph $A \oplus B$ is defined on the same set of vertices as in $A$ and $B$. Moreover, an edge $e$ appears in $A \oplus B$ if and only if $e$ is in exactly one of $A$ or $B$.
	We decompose $A \oplus B$ into edge-disjoint paths; this decomposition is not unique. 
	Note that in $A \oplus B$, goods have degree zero, one, or two. 
	For a good of degree two, the two incident edges belong to the same path. For an agent $i$, let $a_i$, respectively $b_i$, be the number of $A$-edges, respectively $B$-edges, incident to $i$ in $A \oplus B$. Then we have $\min(a_i,b_i)$ alternating paths passing through $i$,  $\max(0, a_i - b_i)$ alternating paths starting in $i$ with an edge in $A$, and $\max(0, b_i - a_i)$ alternating paths starting in $i$ with an edge in $B$. The paths in the decomposition are maximal in the sense that no path can be extended without breaking another one. 

    Let $P$ be a heavy alternating path with respect to $A$ connecting two agents $i$ and $j$ with the edge of $P$ incident to $i$ in $\EH \setminus A$ and the edge incident to $j$ in $\AH$. Then 
	$A \oplus P$ contains the same number of heavy edges as $A$, i.e., $\modulus{\AH} = \modulus{(A \oplus P)^H} = \modulus{\AH \oplus P}$.
    Moreover, the heavy degree of $i$ increases by one, the heavy degree of $j$ decreases by one, and all other heavy degrees are unchanged.  
	
	% \begin{example} \rm
	%     Consider the example shown in Figure \ref{fig1}a and allocation $A$ with 
	%     $$A^H = \{(1,g_1), (2, g_2)\}$$
	%     shown in Figure \ref{fig2}b. Let $B$ be another allocation for which 
	%     $$B^H = \{(1,g_1),(1,g_2)\}.$$
	%     Figure \ref{figPathDecompose} shows $A^H \oplus B^H$. Black edges are only in $A^H$ and dashed edges are only in $B^H$. The path decomposition of $A \oplus B$ consists of the unique path $P=(1,g_2,2)$ in $A \oplus B$.
        % \hfill $\blacksquare$
	% \end{example}
	
	% \begin{figure}
	% \centering
	%     \begin{tikzpicture}[every node/.style={circle,draw,font=\sffamily\small\bfseries, inner sep = 3pt}]
            
        %     \Vertices
            
        %     % the edges
        %     (a1) edge [bend right] node (g1)
        %     \draw[-, very thick, dashed] (a1) -- (g2);
        %     \draw[-, very thick] (a2) -- (g2);
            
        % \end{tikzpicture}

	%     \caption{$A^H \oplus B^H$}
	%     \label{figPathDecompose}
	% \end{figure}	

We will often compare distinct allocations and use a notion of distance between them.
	    The {\em distance} between two allocations is the number of edges that only exist in one of the allocations;
	    formally, the distance between two allocations $A$ and $B$ is $|A \oplus B|$.

It will turn out that in half-integral instances, dealing with alternating paths is not sufficient to compute optimal allocations. For this reason, in Section~\ref{sec:halfInteger} below, we will also rely on a more involved path structure to deal with this problem, namely, \emph{improving walks}.

In the rest of the paper, we denote by $\OPT$ an allocation that maximizes the \NSW\ and by $\OPT^H$ its heavy part. Furthermore, we specify the bundle of the agent $i$ in $\OPT$ by $\OPT_i$ and denote its heavy part by $\OPT^H_i$.

% \subsection{Math Preliminaries}\label{app: basic}

% The following Lemma is useful for showing that certain re-allocations  increase the \NSW. 

% \begin{lemma}\label{basic} \renewcommand{\theenumi}{\alph{enumi})}
% Let $a$, $b$, $c$, $d$, $d_1$, and $d_2$ be non-negative reals. 
% 	\begin{enumerate}
% 		\item If $a \ge b$ and $d \in [0,a-b]$ then 
% 		$ab \le (a-d)(b + d)$ with equality if and only if $d = 0$ or $d = a - b$.
% 		\item If $a \ge b \ge c$, $b \ge c + d_2$, and $a \ge c + d_1 + d_2$ then $abc \le (a - d_1)(b - d_2)(c + d_1 + d_2)$ with equality if and only if $c = 0$ and $d_2 = b$ or $d_2 \in \{0,b - c\}$ and $d_1 \in \{0, a - c - d_2\}$. 
% 	\end{enumerate}
% \end{lemma}
% \begin{proof} For a) we have $a \ge b + d \ge b$ and $d \ge 0$ and hence 
%   \[ (a -d)(b + d) - ab = (a - b - d) d \ge 0\]
%   with equality if and only if $d = 0$ or $d = a- b$.

% Part b) is obvious if $a = 0$. Note that $a = 0$ implies $b = c = d_1 = d_2 = 0$. It is also obvious, if $c = 0$ and $d_2 = b$. Then LHS and RHS are zero. So assume $a > 0$ and either $c > 0$ or $d_2 < b$. In either case $b - d_2 > 0$. We apply part a) twice and obtain
%   \begin{align*}
%     abc &\le a (b - d_2) (c + d_2)    &&\text{equal iff $d_2 = 0$ or $d_2 = b - c $ since $a > 0$.}\\
%         &\le (a - d_1)(b - d_2)(c + d_1 + d_2)  &&\text{equal iff $d_1 = 0$ or $d_1 = a - c - d_2$ since $b - d_2 > 0$.}
% \end{align*}\vspace{-1.5cm}\par
% \end{proof}

\section{Integral Instances}\label{sec:integer}

In this section, we consider \emph{integral} instances, i.e., $q$ is $1$ and $p$ is an integer greater than one. We also term the valuation functions integral. Our main result is a polynomial-time algorithm to find a maximum \NSW\ allocation.

\subsection{Properties of an Optimal Allocation}	
    
We first study the properties of optimal allocations. The main insight is stated in Lemma \ref{heavysim}. Roughly speaking, it states that there exists an optimal allocation $\OPT$, in which heavy goods\footnote{Recall that a good is heavy if it is heavy for some agent.} are assigned as evenly as possible. More formally, the utility profile of $\OPTH$ is lexmin among all heavy-only allocations with the same cardinality. Later, we use this property to prove that the utility profile of $\AH$ at the end of step 1) of Algorithm~\ref{The Algorithm} is equal to the utility profile of $\OPT^H$, if $\OPT$ is chosen cleverly among the set of all optimal allocations. After this, it will not be difficult to prove that the utility profiles of $A$ and $\OPT$ match. 

For an allocation $A$, $\min(A) = \min_i \uv_i(A_i)$ denotes the minimum utility of any of its bundles.

\begin{claim}\label{no-path}\label{prop3}Let $\OPT$ be an optimal allocation.
  \begin{enumerate}
    \item Let $j$ be an agent.
      If $\uv_j(\OPT_j) \ge \min(\OPT) + 2$, then all goods in $\OPT_j$ are heavy for $j$.  As a consequence, only bundles of utility $\min(\OPT)$ and $\min(\OPT) + 1$ can contain light goods. Bundles with higher values only contain goods that are heavy for their owner.
    \item If there is a heavy alternating path with respect to $\OPT$ that starts with an $\OPT$-edge from an agent $i$ to an agent $j$, then $v_i(OPT_i) \le v_j(OPT_j) + p$.
     \item  If a good $g$ is allocated as a light good to an agent $i$ but could be allocated as a heavy good to an agent $j$ who is allocated good $g'$ which is light for $j$, then the allocation is not optimal. 

      \end{enumerate}
	\end{claim} 
	\begin{proof}
          \begin{enumerate}
          \item Assume otherwise. Take a good that is light for $j$ and reallocate it to an agent $i$ for which $\uv_i(\OPT_i) = \min(\OPT)$. This improves the $\NSW$, a contradiction.
           \item Assume such an alternating path exists and call it $P$.
          In $\OPT \oplus P$, $i$ is incident to one fewer heavy edge, and $j$ is incident to one more heavy edge, and, hence, the $\NSW$ changes by the factor
          \[(v_i(\OPT_i) - p)(v_j(\OPT_j) + p)/ v_i(OPT_i) v_j(OPT_j).\]
          This factor must be no larger than one. Thus, $v_i(OPT_i) \le v_j(OPT_j) + p$.
        \item   Swapping the goods $g$ and $g'$ among agent $i$ and agent $j$ increases the value of agent $j$ by $p -1$ and does not decrease the value of agent $i$.
          \end{enumerate}
	 \end{proof}
		
	% \begin{corollary}\label{light-bundles}
	% 	Let $\OPT$ be any optimal allocation. Only bundles of utility $\min(\OPT)$ and $\min(\OPT) + 1$ can contain light goods. Bundles with higher values only contain goods that are heavy for their owner. 
	% \end{corollary} 
	
	% \begin{claim}\label{no-path}
	%     If there is a heavy alternating path with respect to $\OPT$ starting with an $\OPT$-edge from an agent $i$ to an agent $j$ then $v_i(OPT_i) \le v_j(OPT_j) + p$.
	% \end{claim}
	% \begin{proof} Assume such an alternating path exists and call it $P$.
        %   In $\OPT \oplus P$, $i$ is incident to one fewer heavy edge, and $j$ is incident to one more heavy edge, and, hence, the $\NSW$ changes by the factor
        %   \[(v_i(\OPT_i) - p)(v_j(\OPT_j) + p)/ v_i(OPT_i) v_j(OPT_j).\]
        %   This factor must be no larger than one. Thus $v_i(OPT_i) \le v_j(OPT_j) + p$.
	%   \end{proof}
	
	% \begin{claim} \label{prop3}
	%     If a good $g$ is allocated as a light good to an agent $i$ but could be allocated as a heavy good to an agent $j$ who is allocated good $g'$ which is light for $j$, then the allocation is not optimal. 
	% \end{claim}
	% \begin{proof}
	%     Swapping the goods $g$ and $g'$ among agent $i$ and agent $j$ increases the value of agent $j$ by $p -1$ and does not decrease the value of agent $i$.
	%  \end{proof}
	
%The rest of this section is dedicated to proving the following lemma. 
	
	\begin{lemma} \label{heavysim} Let $K \in \NN$ and assume that there is an allocation that maximizes $\NSW$ and has exactly $K$ goods allocated as heavy. Then, there is an optimal allocation $A$ such that the utility profile of $\AH$ is lexmin among all heavy-only allocations with exactly $K$ goods allocated as heavy.
	\end{lemma}
	\begin{proof}
We choose $A$ and heavy-only $\CH$ as follows: (1) $A$ is an optimal allocation with exactly $K$ goods allocated as heavy, (2) $\CH$ is lexmin among all allocations of $K$ heavy goods, and (3) the distance between $\AH$ and $\CH$ is minimum among all allocations satisfying (1) and (2). 
	
    We will show that $A^H$ and $C^H$ agree. Assume otherwise, and let us consider $\AH \oplus \CH$. We label an edge with either $A$ or $C$, indicating whether it belongs to $\AH$ or $\CH$. Note that in this graph, goods have degree zero, one, or two. $\AH \oplus \CH$ decomposes into edge-disjoint maximal alternating paths and cycles. We first show that there are no heavy alternating cycles.
    
    \begin{observation}\label{noCycle}
    There are no cycles in $\AH \oplus \CH$.
    \end{observation}
    \begin{proof}
    Assume first that there is an alternating cycle, say $D$. Then $\CH \oplus D$ has the same utility profile as $\CH$ and is closer to $\AH$, a contradiction. 
     \end{proof}
    
    Hence, we have only alternating paths. We next make more subtle observations about the edge-disjoint alternating paths in $\AH \oplus \CH$. First, we show that we cannot have an alternating path with both endpoints as goods. 
    
    \begin{observation}\label{noEvenLGoods}
    There are no maximal paths in $\AH \oplus \CH$ with both endpoints as goods.
    \end{observation}
    \begin{proof}
    Assume otherwise, and let $P$ be such a path. Then, $\CH \oplus P$ has the same utility profile as $\CH$ and is closer to $\AH$, a contradiction. 
     \end{proof}

    So at least one endpoint of each maximal alternating path is an agent. If there is an even length maximal alternating path, both endpoints are agents. Let $P$ be such a path, let $i$ and $j$ be its endpoints, and assume w.l.o.g.~that $P$ starts in $i$ with an edge in $\AH$ and ends in $j$ with an edge in $\CH$. Then, $\abs{\AH_i} > \abs{\CH_i}$ and $\abs{\CH_j} > \abs{\AH_j}$. 
    Set $Q$ to the empty path. 
    
    If all maximal alternating paths have odd length, exactly one endpoint of each path is an agent. Let $i$ and $j$ be agents with $\abs{\AH_i} > \abs{\CH_i}$ and $\abs{\CH_j} > \abs{\AH_j}$, respectively, and let $P$ and $Q$ be maximal alternating path starting in $i$ and $j$, respectively. The other endpoints of $P$ and $Q$ are goods. 
    
    We next show $\abs{\AH_j} \leq \abs{\AH_i} -2$.

	\begin{observation}\label{jHasLessHeavyGoodThani}
	$\abs{\AH_j} < \abs{\AH_i}$.
	\end{observation}	
	\begin{proof}
		Assume otherwise, i.e.,  
		$\abs{\AH_j} \ge \abs{\AH_i}$.
		%$u'_j(A_j) \ge \uv_i'(A_i)$. 
		Then,
		$\abs{\CH_j} > \abs{\CH_i}+1$,
		%$u'_j(C_j) \ge u'_i(C_i) - 2$ 
		and, hence, $\CH \oplus P \oplus Q$ is lexicographically smaller than $\CH$, a contradiction.
	 \end{proof}
		
	\begin{observation}\label{jHasLessHeavyGoodThani-1}
	$\abs{\AH_j} < \abs{\AH_i} -1$.
	\end{observation}
	\begin{proof}
	If 
		$\abs{\AH_j} = \abs{\AH_i} -1$,
		%$u'_j(A_j) = \uv_i'(A_i) - 1$, 
		then $\AH\oplus P \oplus Q$ and $\AH$ have the same utility profile with respect to heavy goods. Also, $\AH \oplus P \oplus Q$ is closer to $\CH$ than $\AH$. Finally, we swap the goods that are light for $i$ in $A_i$ with the goods that are light for $j$ in $A_j$. The value of the resulting bundles for $i$ and $j$ are at least $\uv_j(A_j)$ and $\uv_i(A_i)$, respectively. If either inequality is strict, $A$ was not optimal, a contradiction. So, both inequalities are equalities, and, thus, the resulting allocation is again optimal, and with respect to heavy edges, it has the same utility profile as before and is closer to $\CH$, a contradiction.
	 \end{proof}
		
	\begin{observation}\label{noEvenLAgents}
	The paths $P$ and $Q$ do not exist. 
	\end{observation}	
	\begin{proof}
		By Observation \ref{jHasLessHeavyGoodThani-1}, we have $\abs{\AH_j} \le \abs{\AH_i} -2$. By Claim~\ref{no-path}.2, $v_i(A_i) \le v_j(A_j) + p$, and, hence, $A_j$ contains $p$ goods light for $j$. Augmenting $P$ and $Q$ to $A$ and moving these light goods to $A_i$ yields an allocation that has the same NSW as $A$ and is closer to $C^H$, contradicting the choices of $A$ and $C^H$.
 \end{proof}

We can now complete the proof of Lemma~\ref{heavysim}. Since $P$ and $Q$ do not exist, $A^H = C^H$.  \end{proof}

	\newcommand{\LG}{\mathit{LG}}

    Example~\ref{ex:int-halfInt} shows that Lemma \ref{heavysim} does not hold when $s$ is half-integral.

\subsection{Algorithm}

    Algorithm ~\ref{The Algorithm} operates in three phases. 
    
    The first phase (step 1) finds a heavy-only allocation that maximizes the $\NSW$. This phase is equivalent to maximizing the $\NSW$ in a binary instance. Barman et al. \cite{BarmanKV18AAMAS} proved that this is possible in polynomial time. Notice that, after this phase, light goods remain unallocated.
    
    In the second phase (step 2), we greedily allocate the light goods (one by one) to an agent with minimum utility. As such, we term this phase ``allocating light goods''.
    
    In the third phase (step 4), we try to improve the NSW by re-allocating heavy goods to agents considering them light. More precisely, we take a heavy good from the bundle of an agent with maximum utility and allocate it to an agent with minimum utility as long as the $\NSW$ increases. In Lemma \ref{importantLemma} below we show that the reallocated goods are light for their new owners. This means that as long as there is progress, we turn some heavy good into a light good.

	% \begin{algorithm}[t]
        % \DontPrintSemicolon
	% 	\caption{TwoValueMaxNSW} \label{alg2}
	% 	Input : $N, M, \uv = (\uv_1, ... , \uv_n)$ \;
	% 	Output: allocation $A$\;
  
	% 	%\begin{algorithmic}[]
        % \medskip
	% 	  \tcc{Phase 1: Heavy-Only Allocation}
        % let $\uv'_i: 2^M \rightarrow \mathbb{N}$ be an additive function with $\uv'_i(g) = 1$ if $\uv_i(g) = p$ and $\uv'_i(g) = 0$ otherwise for all $i \in N$, $g \in M$\;
        % $\uv' \leftarrow (\uv'_1, \dots , \uv'_n)$\;
        % $A =$ BinaryMaxNSW$(N, M, \uv')$\;

        % \medskip
        % \tcc{Phase 2: Allocating Light Goods}
        % number the agents so that $\uv_1(A_1) \le \uv_2(A_2) \le ... \le \uv_n(A_n)$\;
        % \While {there is an unallocated good $g$}{
        %     \tcc{$\uv_1(A_1) \le \uv_2(A_2) \le ... \le \uv_n(A_n)$}
        %     let $k$ be the maximum index s.t.~$\uv_k(A_k) = \uv_1(A_1)$\;
        %     $A \leftarrow A \cup \{(k,g)\}$
        % }

        % \medskip
        % \tcc{Phase 3: Increasing $\NSW$} 
	% 	\While{$ \uv_n(A_n) >  p \cdot \uv_1(A_1) + p $ }{
        %     let $k$ be the maximum index s.t.~$\uv_k(A_k) = \uv_1(A_1)$\;
        %     let $t$ be the minimum index s.t.~$\uv_t(A_t) = \uv_n(A_n)$\;
        %     let $g$ be a good such that $(t,g) \in A$\;
        %     $A \leftarrow A \setminus \{(t,g)\}$\;
        %     $A \leftarrow A \cup \{(k,g)\}$
        % }
		
	% 	return $A$
	% \end{algorithm}

    \subsubsection{Phase 1: Heavy-Only Allocations}
    % \paragraph{Phase 1:} Heavy-Only Allocations
    In the first step, we compute a heavy-only allocation that maximizes the $\NSW$. For completeness, let us recapitulate the algorithm from Barman et al.~\cite{BarmanKV18AAMAS}. 
	
	In order to compute a heavy-only allocation that maximizes $\NSW$, we start with a heavy-only allocation $A$ of maximum cardinality, i.e., in $A$ every heavy good is assigned to an agent for which it is heavy. We then improve the $\NSW$ of $A$ by augmentation of some heavy alternating paths. We search for a heavy even-length alternating path $P$ that connects an agent $i$ to an agent $j$, starting with an edge outside $A$ and ending with an edge in $A$, and with the heavy degree of $j$ at least two larger than the heavy degree of $i$ in $A$. Given any such path $P$, we augment $P$ to $A$, i.e., we update $A$ to $A \oplus P$. When no such path can be found, the algorithm stops and returns allocation $A$. 

	% \begin{algorithm} 
	% 	\caption{BinaryMaxNSW} \label{alg1}
	% 	Input : $N, M, v = (\uv_1, ... , \uv_n)$ 
		
	% 	Output: allocation $A$
		
	% 	%\begin{algorithmic}
	% 		 let $G$ be the corresponding graph, i.e, agent $i$ is connected to good $g$ iff $g$ is heavy for $i$, and let $A$ an arbitrary assignment of the items to the agents
			
        %         \While{there is an alternating path $a_0, g_1, ... , g_k, a_k$ such that ~$|A_0| \le |A_k| - 2$}{
                
	% 		        \For {$\ell \leftarrow k$ to $1$}{
           
	% 		         $A \leftarrow A \backslash (a_{\ell}, g_{\ell})$\\
	% 		         $A \leftarrow A \cup (a_{\ell-1}, g_{\ell})$
        %             }
	% 		%\EndFor		
	% 		%\EndWhile
        %         }
                
	% 	  return $A$
	% 	%\end{algorithmic}
		
	% \end{algorithm}
	
	Barman et al.\ \cite{BarmanKV18AAMAS} proved that this algorithm returns an allocation with maximum $\NSW$ for binary instances. Furthermore, Halpern et al.\ \cite{HalpernPPS20} proved that in binary instances the set of lexmin allocations\footnote{We defined the lexmin allocation as the allocation where the maximum is smallest, subject to this the second largest is smallest, and so on. Alternatively, one could ask for the allocation where the minimum is largest, subject to this the second smallest is largest, and so on. In general, these definitions are not equivalent. Imagine, we are to choose between the tuples $(3,4,8)$ and $(2,6,7)$. Both have sum 15 and the two definitions select different tuples. However, in our context the two definitions are equivalent. Assume otherwise, say the first definition yields the tuple $(a_1,\ldots,a_n)$ and the second definition yields $(b_1,\ldots,b_n)$. Then in the bipartite graph corresponding to the two allocations, there is an alternating path that allows to move the allocations towards each other, demonstrating that neither is optimal.} is identical to the set of allocations with maximum $\NSW$.
	
	\begin{lemma}[\cite{BarmanKV18AAMAS,HalpernPPS20}]\label{binaryNSW}
	The first step of Algorithm~\ref{The Algorithm} computes an allocation $A^H$ with maximum $\NSW$. Furthermore, the optimal allocations of the heavy items are exactly the lexmin allocations.
      \end{lemma}

         \subsubsection{Phase 2: Adding Light Goods}\label{phase2}
         Phase 2 is trivial. We allocate the light goods greedily, i.e., we allocate the light goods one by one and always to an agent of minimum utility. The following Lemma is more general than needed for the correctness of Phase 2. It states that for a fixed allocation of the heavy goods, adding the light goods greedily maximizes NSW. We will need the increased generality for the analysis of Phase 3.

         \newcommand{\xmin}{x_{\min}}

         	\begin{lemma} \label{greedy} 
 Let $\OPT$ be any optimal allocation, and let $L$ be the set of goods that are allocated as light goods in $\OPT$. Then, the following allocation is also optimal. Start with $\OPTH$ and then allocate the goods in $L$ greedily, i.e., allocate the goods one by one and, for each good $g \in L$, choose an arbitrary agent $i$ that currently owns a bundle of minimum value, and assign $g$ to them.
        \end{lemma}
	\begin{proof} Let $\xmin$ be the minimum value of any bundle in $\OPT$. Then all light goods are in bundles of value $\xmin$ and $\xmin + 1$. Let $I$ be the indices of the bundles in $\OPT$ with heavy weight at most $\xmin$, and let $k$ be the number of bundles in $I$ with weight $\xmin + 1$. Then $\abs{L} + \sum_{i \in I} v_i(\OPT_i^H) = \abs{I} \xmin + k$. 

          Assume first that the goods in $L$ are light for all agents. Greedy will first fill up all bundles in $I$ to value $\xmin$. This requires $\abs{I}\xmin - \sum_{i \in I} v_i(\OPT_i^H) = \abs{L} - k$ items. Then it will increase the value of $k$ bundles of value $\xmin$ to $\xmin + 1$. Thus Greedy will produce $k$ bundles of value $\xmin + 1$ and $\abs{I} - k$ bundles of value $\xmin$. Therefore, 
Greedy will generate an allocation with the same $\NSW$ as $\OPT.$

If the goods in $L$ are not light for all agents, assume for a moment that they are light for all agents, and proceed as above. At the end, reconvert the value of any good in $L$ to its true value for its owner. This cannot decrease NSW. If it increases NSW, $\OPT$ was not optimal.
\end{proof}

         \subsubsection{Phase 3: Conversion of Heavy Goods into Light Goods}\label{sec-phase3}

         In Phase 3, we convert some of the heavy goods to light goods, see Algorithm~\ref{phase3}. As long as $ \uv_n(A_n) >  p \cdot \uv_1(A_1) + p $, we take a good (guaranteed to be heavy) from one of the heaviest bundles (we take the heaviest bundle with smallest index) and move it to a bundle of minimum value (we take the lightest bundle with largest index). The moved good will be light for the receiving agent. The change increases the NSW. 

         \begin{algorithm}[t]
               %\tcc{Phase 3: Increasing $\NSW$} 
		\While{$ \uv_n(A_n) >  p \cdot \uv_1(A_1) + p $ }{
            let $k$ be the maximum index s.t.~$\uv_k(A_k) = \uv_1(A_1)$\;
            let $t$ be the minimum index s.t.~$\uv_t(A_t) = \uv_n(A_n)$\;
            let $g$ be a good in $A_t$. Remove it from $A_t$ and add it to $A_k$.}
              \caption{\label{phase3} The algorithm for Phase 3. The bundles are ordered initially such that $v_1(A_1) \le v_2(A_2) \le \ldots \le v_n(A_n)$. The choices of $k$ and $t$ guarantee that this property is maintained. }
              \end{algorithm}

              \subsubsection{Correctness}
Phase 1 already gives us an optimal allocation $X$ of the heavy goods, which is a lexmin allocation of the heavy goods. In Phase 2, we allocate the light goods to the least valuable bundles and obtain an allocation that is optimal amongst all allocations in which all heavy goods are allocated as heavy goods. In Phase 3, we take heavy goods from heaviest bundles and move them to lightest bundles as long as such a move improves NSW. Let $A$ be the allocation computed by Algorithm \ref{phase3}. First, we prove that there is an allocation $\OPT$ with maximum $\NSW$ such that the utility profile of $\OPT^H$ and $\AH$ are the same. Then we prove that in allocation $\OPT$, the remaining goods are allocated the same way as in $A$. 
    We start by establishing some invariants of Algorithm \ref{phase3}.

 \begin{lemma}\label{importantLemma} Fix a numbering of the agents at the beginning of Phase 3 such that $\uv_1(A_1) \le \uv_2(A_2) \le \ldots \le \uv_{n-1}(A_{n-1}) \le \uv_n(A_n)$. During Phase 3, the following holds:
  \renewcommand{\theenumi}{\textrm{\alph{enumi}}}

  \begin{enumerate}
    \item \label{neatOrder} The ordering $\uv_1(A_1) \le \uv_2(A_2) \le \ldots \le \uv_{n-1}(A_{n-1}) \le \uv_n(A_n)$ is maintained.
    \item \label{oneStep} Let $i$ be any agent. If $A_i$ contains a good that is light for $i$, then $\uv_i(A_i) \le \uv_1(A_1) + 1$. 
    \item \label{lexmin} $\AH$ is lexmin among all heavy-only allocations of the same cardinality. 
    \item Whenever a good is moved in Phase 3, say from bundle $A_t$ to bundle $A_k$, all goods in $A_t$ are heavy for $t$ and light for $k$.
    \item \label{NSWincreases}  Each iteration of the while-loop increases the $\NSW$.
  \end{enumerate}
\end{lemma}
\begin{proof} We prove statements a) to d) by induction on the number of iterations in Phase 3. Before the first iteration, a) and d) trivially hold. Claim b) holds since in Phase 2 we allocate only goods that are light for every agent, and since the next good is always added to a lightest bundle. Claim c) holds by Lemma \ref{binaryNSW}.

  Assume now that a) to c) hold before the $i$-th iteration, and that we move a good $g$ from $A_t$ to $A_k$ in iteration $i$. We will show that d) holds for $A_t$ and $A_k$, and that a) to c) hold after iteration $i$.

  By the condition of the while-loop, we have $\uv_t(A_t) > p \cdot (\uv_k(A_k) + 1)$. Thus, $A_t$ contains only goods that are heavy for $t$ by part b) of the induction hypothesis. Let $g$ be any good in $A_t$. If we also had $\uv_k(g) = p$, then moving $g$ from $A_t$ to $A_k$ would result in an allocation of heavy goods that is lexicographically smaller, a contradiction to c). Thus, $\uv_k(g) = 1$.

  Let $c_t$ be the maximum number of heavy agents owned by any agent. Note that agent $t$ owns $c_t$ heavy goods. Let $S$ be the set of agents that either own $c_t$ items in $\AH$ or to which there is a heavy alternating path starting with an edge in $\AH$ from an agent owning $c_t$ heavy items. Let $s = \abs{S}$. Since $\AH$ is lexmin, all agents in $S$ own $c_t - 1$ or $c_t$ heavy items, i.e., $n - t + 1$ agents own $c_t$ and the remaining own $c_t - 1$. No agent outside $S$ considers a good owned by an agent in $S$ heavy, because such an agent would then also belong to $S$. Hence, all the heavy items assigned to the agents in $S$ must be assigned to them in any allocation. After the conversion of a heavy item to a light item, the number of heavy items that must be assigned to the agents in $S$ is reduced by at most one; the conversion of $g$ reduces the number by one. The lexmin-assignment of these items is to have one less bundle of size $c_t$ and one more of size $c_t - 1$. Thus the heavy allocation after the conversion is again lexmin.
		
	Since $k$ is the largest index such that $\uv_k(A_k) = \uv_1(A_1)$ before the $i$-th iteration, b) holds after the $i$-th iteration.

  It remains to show that part a) holds after the $i$-th iteration. The value of the $k$-th bundle increases by 1, and the value of the $t$-th bundle decreases by $p$. We need to show $\uv_t(A_t) -p \ge \uv_{t-1}(A_{t-1}) +\delta$, where $\delta =1 $ if $k = t-1$ and $\delta = 0$ otherwise. 
  
  \begin{itemize}
    \item If $k = t-1$, we have 
    $\uv_t(A_t) \ge p \cdot (\uv_{t-1}(A_{t-1} )+ 1) + 1$, and, hence, $\uv_t(A_t) - \uv_{t-1}(A_{t-1}) - p - 1 \ge (p-1) \cdot \uv_{t-1}(A_{t-1})  \ge 0$. 
    \item If $k < t-1$, by definition of $t$, $\uv_t(A_t) > \uv_{t-1}(A_{t-1})$. If all goods in $A_{t-1}$ are heavy for $t-1$, the difference in weight is at least $p$ and we are done. 
    If $A_{t-1}$ contains a good that is light for $t-1$, then $\uv_{t-1}(A_{t-1}) \le \uv_k(A_k) + 1$ by condition b), and, hence, $\uv_t(A_t) \ge p \cdot \uv_{t-1}(A_{t-1}) + 1$. This implies $\uv_t(A_t) \ge \uv_{t-1}(A_{t-1}) + p$ except if $\uv_{t-1}(A_{t-1}) = 0$. In the latter case, $k = t-1$, but we are in the case $k < t-1$.

  We also need to show that after moving the good, $\uv_k(A_k) \leq \uv_{k+1}(A_{k+1})$. By the choice of $k$, $\uv_k(A_k) \leq \uv_{k+1}(A_{k+1}) - 1$ holds before moving the good. After moving the good, by condition d), $\uv_k(A_k)$ increases by $1$, and, therefore, $\uv_k(A_k) \leq \uv_{k+1}(A_{k+1})$.
  \end{itemize}

  Let us finally show \ref{NSWincreases}). By d) we know that whenever a good is moved from bundle $A_t$ to bundle $A_k$, all goods in $A_t$ are heavy for $t$ and light for $k$.
  Therefore, moving a good from $A_n$ to $A_1$ increases the $\NSW$ if and only if $(\uv_n(A_n) - p)(\uv_1(A_1) + 1) > \uv_n(A_n)\uv_1(A_1)$, which holds true if and only if $ \uv_n(A_n) >  p \cdot \uv_1(A_1) + p $.
   \end{proof}

	\begin{lemma}\label{oneCase}
	Let $A$ be the output of Algorithm \ref{phase3}. Let $\OPT$ be an allocation that maximizes the $\NSW$ and, subject to that, maximizes $\abs{\OPTH}$. Then, $\abs{\OPTH} \ge \abs{\AH}$. 
	\end{lemma}

	\begin{proof} Among the allocations that maximize $\NSW$ and subject to that maximize $\abs{\OPTH}$, let $\OPT$ be the lexmin allocation. Assume $\abs{\OPTH} < \abs{\AH}$. By the choice of $\OPT$, $A$ cannot maximize the $\NSW$ (since $\abs{\OPTH}$ is maximum for any allocation that maximizes the $\NSW$). We first show that we may assume $\OPTH_i \subseteq \AH_i$ for all $i$. By Lemma~\ref{importantLemma}c), $\AH$ is lexmin among all heavy-only allocations of cardinality $\abs{\AH}$. We obtain a lexmin heavy-only allocation $\CH$ of cardinality $\abs{\OPTH}$ from $\AH$ by repeatedly removing a good from the lowest indexed bundle of maximum utility. Since $\OPTH$ is a lexmin heavy-only allocation of cardinality $\abs{\OPTH}$ (Lemma~\ref{heavysim}), the utility profiles of $\CH$ and $\OPTH$ agree, and, hence, there is a bijection $\pi$ of the set of agents such that $\abs{\CH_i} = \abs{\OPTH_{\pi(i)}}$. 
		Let $\ell_i$ be the number of goods in $\OPT_{\pi(i)}$ which are light for $\pi(i)$. Note that the number of goods that are not allocated under $C^H$ is equal to the total number of  goods that are light to their owner under $\OPT$, i.e, $\Sigma_{i \in [n]} \ell_i$. 
		Obtain an allocation $C$ from $\CH$ by giving $\ell_i$ not yet allocated goods to $C_i$. Then, $\uv_i(C_i) \ge \uv_{\pi(i)}(OPT_{\pi(i)})$ for all $i$. Thus, $C$ is optimal, and $\uv_i(C_i) = \uv_{\pi(i)}(OPT_{\pi(i)})$. Also, $C^H_i \subseteq \AH_i$ for all $i$. So, choosing $\OPT$ as $C$, we may assume 
        $\OPTH_i \subseteq \AH_i$ for all $i$. 
        
		%Since $C$ is optimal, \Cref{greedy} gives us an alternative way of obtaining an optimal allocation. Start from $\CH$ and then allocate the goods that are allocated as light goods (i.e., the goods that are light for their owner) in $\OPT$ in a greedy fashion. Let then $R$ be the set of agents from which we removed a good in moving from $\AH$ to $\CH$. Note that no good is added to the bundle of an agent in $R$ when adding goods greedily. Otherwise, we would have a contradiction to \Cref{prop3}.

  Let $t$ be such that $A^H_t \setminus \OPT_t^H \not= \emptyset$, and let $g \in A^H_t \setminus \OPT_t^H$. In $\OPT$, $g$ is allocated to some agent $j$ as a light good. Note that $\OPT_t$ is heavy-only. Otherwise, interchange the light good in $\OPT_t$ with $g$ and re-convert $g$ to a heavy good, thus improving the NSW of $\OPT$, a contradiction. Hence, $\uv_t(\OPT_t) + p \leq \uv_t(A_t) \leq \uv_n(A_n)$. 
  
  Since $A$ is not optimal, there is an agent $k$ such that $\uv_k(\OPT_k) > \uv_k(A_k)$. Since $\OPTH_k \subseteq A^H_k$, the bundle $\OPT_k$ contains a good that is light for $k$. 
		
Since $A$ is the output of Algorithm~\ref{phase3}, the while condition of Algorithm~\ref{phase3} (line 1) does not hold,
		and, hence,
		\[ \uv_n(A_n) \leq p \cdot \uv_1(A_1) + p. \]
Therefore, we get
	\[	\uv_t(\OPT_t) + p \le \uv_n(A_n) \le p \cdot \uv_1(A_1) + p \le p \cdot \uv_k(A_k) + p \le p \cdot \uv_k(\OPT_k), \]
and, hence,
		\[\uv_t(\OPT_t) \leq p\cdot \uv_k(\OPT_k) - p.\]
		Moving a light good from $k$ to $j$ and taking $g$ from $j$'s bundle and allocating it as a heavy good to $t$ increases the number of heavy goods allocated to $t$ by one. This reallocation of goods does not decrease the $\NSW$ as
  \[ (\uv_k(\OPT_k) - 1)(\uv_t(\OPT_t) + p) - \uv_k(\OPT_k)\uv_t(\OPT_t)
        = p \cdot \uv_k(\OPT_k) - \uv_t(\OPT_t) - p
        \ge 0. \]
		For the new allocation $\widehat{\OPT}$, we have 
            \[
                \NSW (\widehat{\OPT}) \geq \NSW (\OPT) \hspace{0.5cm} \text{ and } \hspace{0.5cm} \abs{\widehat{\OPT}^H} > \abs{\OPTH}, \]
   a contradiction to the choice of $\OPT$.
	 \end{proof}

	Phase 3 starts with all heavy goods allocated as heavy and ends with an allocation $A$ with $\abs{A^H} \le \abs{\OPTH}$ for some optimal allocation $\OPT$ (Lemma~\ref{oneCase}). Therefore, we have in some round of Phase $3$ an allocation $\Tilde{A}$ with $\abs{\Tilde{A}^H} = \abs{\OPTH}$ for some optimal allocation $\OPT$. Let $K = \abs{\Tilde{A}^H}$. By Lemma~\ref{importantLemma}.\ref{lexmin}, $\Tilde{A}^H$ is lexmin among all heavy-only allocations with the same cardinality, and by Lemma~\ref{heavysim} we may assume the same for $\OPTH$. We get the following Corollary.
	
	\begin{corollary}\label{meetingOPT}
	   % There is an optimal allocation $\OPT$ such that $\Tilde{A}^H = \OPTH$.
        There is an optimal allocation $\OPT$ such that $\Tilde{A}^H$ and $\OPTH$ have the same utility profile.
	\end{corollary}
	
        So far, we have proven that considering only heavy allocated goods in $\Tilde{A}$ and $\OPT$, we end up having the same utility profile. 
        %By \Cref{greedy} and Lemma \ref{importantLemma}.\ref{neatOrder}, in both allocations $\OPT$ and $\Tilde{A}$, light goods are allocated as evenly as possible. So we can conclude that the utility profiles of $\Tilde{A}$ and $\OPT$ are equal. 
        Light goods are allocated in a greedy manner during Phases $2$ and $3$. 
        By Lemma~\ref{greedy} we conclude that $\Tilde{A}$ has maximum $\NSW$. 
        In each round of Phase $3$, the \NSW\ increases, so $\NSW (A) \ge \NSW (\Tilde{A}) = \NSW (\OPT)$.

\IntegerResult*

	\begin{proof}
	   We have already proven that the algorithm outputs an allocation that maximizes \NSW. It remains to prove that it runs in polynomial time. The first phase runs in polynomial time~\cite{BarmanKV18AAMAS}. The second phase clearly takes polynomial time. By Lemma~\ref{importantLemma}.d, the number of heavy goods under $A$ is decreasing after each iteration of the third phase. Therefore, the number of iterations is bounded by the number of heavy goods. Overall, the algorithm terminates in polynomial time. 
\end{proof}

\section{Half-Integral Instances}\label{sec:halfInteger}

We give a polynomial time algorithm for half-integral instances, i.e., $q$ is equal to $2$ and $p$ is an odd integer greater than two. Again, we will first deal with the case that all heavy goods must be allocated as heavy goods (Section~\ref{heavy as heavy}). The algorithm consists of three steps. We first determine the lexmin allocation of the heavy goods and then allocate the light goods greedily. This is as in the integer case.
Let $x$ be the minimum value of any bundle in the resulting allocation. Then only bundles of value $x$, $x + \sfrac{1}{2}$ and $x + 1$ can contain light goods. Call these bundles the small bundles. We optimize the allocation of the small bundles by maximizing the number of bundles of value $x + \sfrac{1}{2}$. Essentially, we try to combine a bundle of value $x$ and a bundle of value $x + 1$ to two bundles of value $x + \sfrac{1}{2}$. The vehicle for doing so are \emph{alternating walks}. Alternating walks are sequences of alternating paths that are glued together at their endpoints. We find alternating walks by establishing a connection to matchings with parity constraints~\cite{Akiyama-Kano}. We develop the theory in Section~\ref{APC exists} and the algorithm in Section~\ref{Algorithmics}. For the algorithm, we exploit known algorithms for maximum matchings with parity constraints. 
In Section~\ref{heavy as light}, we then deal with the situation that heavy goods can be allocated as light goods. We repeatedly take a heavy good from a bundle of highest value and allocate it as a light good to a bundle of minimum value. In contrast to integral instances, it will be necessary to reoptimize after each such move.

Let $s = p/q$. Recall that $q = 2$ and $p$ is an odd integer greater than two.

\subsection{Heavy Goods must be Allocated as Heavy Goods}\label{heavy as heavy}

In this section, \emph{we assume that a heavy good, i.e., a good that is heavy for some agent, is allocated to an agent that considers it heavy.} Therefore, the total value of all bundles is the number of light goods plus $s$ times the number of heavy goods. The algorithm for constructing an optimal allocation consists of the first three steps of Algorithm~\ref{The Algorithm}. The first two steps are as in the integer case. In the next subsections, we give more details.
% \begin{enumerate}
% \item 
% Determine the lexmin allocation of the heavy goods, i.e., push heavy goods towards smaller bundles as much as possible. Let $b_n \ge b_{n-1} \ge \ldots \ge b_1$ be the number of heavy goods allocated to the different agents sorted in decreasing order. Then $b_n$ is minimal among all allocations of the heavy goods, and given that $b_n$ has its minimal value, $b_{n-1}$ is minimal, and so on. In other words, the string $b_nb_{n-1}\ldots b_1$ is lexicographically minimal. 

% \item Allocate the light goods greedily, i.e., allocate the light goods one by one and always add the next good to a bundle of smallest value.
% \item Let $x$ be the minimum value of any bundle in the resulting allocation. Call the bundles of value $x$, $x + \half$, and $x + 1$ the small bundles and let $N_s$ be the owners of the small bundles. Optimize the allocation of the small bundles, i.e., allocate the goods contained in the small bundles to the agents in $N_s$ so as to allocate heavy goods only to agents considering them heavy, give each bundle a value in $\sset{x, x + \half, x + 1}$, and maximize the number of bundles of value $x + \half$.
% \end{enumerate}

\subsubsection{LexMin Allocation of Heavy Goods}

We need the concept of an alternating path. Let $A$ be an allocation, let $A^H$ be its restriction to the heavy edges, and let $E^H$ be the set of all heavy edges.  An alternating path (with respect to $A^H$ and $\bar{A} = E^H \setminus A^H$) uses alternatingly edges\footnote{We will later also consider alternating path with respect to $A$ and $O$ that use alternatingly edges in $A^H$ and $\OH$.} in $A^H$ and $\bar{A}$ and has agents as endpoints. The first and the last edge of the path have different types; one belongs to $A$ and one belongs to $\bar{A}$.\footnote{In Section~\ref{altpath}, we considered alternating paths without this restriction. We do not need this generality anymore.} The endpoint incident to the $A$-edge is called the $A$-endpoint of the path. The other endpoint is then the $\bar{A}$-endpoint. Augmentation of an alternating path to $A^H$ effectively moves a heavy good from the $A$-endpoint to the $\bar{A}$-endpoint.  

The lexmin allocation of the heavy goods is readily computed. Start with an arbitrary allocation $A^H$ of the heavy goods. As long as there is an alternating path starting with an $A$-edge from an agent $i$ to an agent $j$ that owns at least two heavy goods less than $i$, augment the path. This decreases the number of heavy goods at $i$ by one and increases the number of heavy goods at $j$ by one.\footnote{Alternatively, one can use parametric flow~(\cite{ParametricFlow, Darwish-Mehlhorn}). Have an edge of capacity $s$ from an agent to each good that the agent considers heavy, and add a source and a sink. Have an edge of capacity $\lambda$ from the source to each agent, and an edge of capacity $s$ from each good to the sink. Determine max-flows for all values of $\lambda$. For $\lambda \le b_1$ the flow is $\lambda n$. For $b_1 \le \lambda \le b_2$ the max-flow is
$b_1 + \lambda(n-1)$, and so on.}

With respect to a lexmin allocation $A^H$, let $\bar{k}$ (called $b_n$ above) be the maximum number of heavy goods assigned to any agent. Define $R'_{\bar{k} + 1} = \emptyset$ and sets $R_\ell$, $R'_\ell$, and $S_\ell$ of agents for $\ell = \bar{k}$, $\bar{k} -1$, $\bar{k} - 2$, and so on:\vspace{-0.8cm}
\begin{align*}
  R_\ell & = \{ a \, ; \, \parbox[b]{0.52\textwidth}{$a$ owns $\ell$ heavy goods and does not belong to $R'_{\ell + 1}$.}\},\\
  R'_\ell &= \{ a \, ; \, \parbox{0.8\textwidth}{$a$ owns $\ell - 1$ heavy goods and there is an alternating path starting with an }\\
         &\parbox{6.5cm}{\hspace{1cm} $A$-edge from an agent in $R_\ell$ to $a$.} \},\\
  S_\ell &= R_\ell \cup R'_\ell.
\end{align*}\vspace{-1.5cm}

Let $S_{\ge k} = \cup_{\ell \ge k} S_\ell$. Define $S_{> k}$, $S_{\le k}$ and $S_{<k}$ analogously.\footnote{The sets $R_\ell$ and $R'_\ell$ are defined with respect to a particular lexmin allocation $A$. Whether an agent belongs to $R_\ell$ or $R'_\ell$ depends on the particular choice of $A$. However, $S_\ell$ is independent of the particular choice. Also, which goods are assigned to a particular agent depends on the choice of $A$. }

\begin{lemma}\label{Sk} For any $\ell$, agents in $S_\ell$ own either $\ell$ or $\ell - 1$ heavy goods. Each of them could own $\ell$ goods via a transfer from another agent in the set. For any $k$: No good considered heavy by an agent in $S_{<k}$ is assigned to an agent in $S_{\ge k}$, and all goods considered heavy by an agent in $S_{\ge k}$ but by no agent in $S_{< k}$ are assigned as a heavy good to an agent in $S_{\ge k}$. In other words, the heavy goods assigned to the agents in $S_{\ge k}$ are exactly the heavy goods that are considered heavy by an agent in $S_{\ge k}$ but by no agent in $S_{< k}$.\end{lemma}
\begin{proof} By construction, each agent in $R'_k$ can be reached from an agent in $R_k$ by an alternating path starting with an $A$-edge. Augmention of this path increases the number of heavy items owned by the former agent.

  We show next that there is no alternating path starting with an $A$-edge from an agent in $S_{\ge k}$ to an agent in $S_{<k}$. Assume otherwise. Let $P$ be a such a path, and let agent $a$ be the start vertex of $P$.  If $a \in R'_\ell$ for some $\ell \ge k$, there is an alternating path $Q$ starting in an agent in $R_\ell$ and ending in $a$. Consider the concatenation of $Q$ and $P$ instead of $P$. So, we may assume that the path starts in an agent in $R_\ell$ for some $\ell \ge k$. If the agent at the other end of the path owns $\ell - 1$ heavy goods, it belongs to $S_\ell$. If the agent at the other end of the path owns at most $\ell - 2$ heavy goods, augmentation of the path to $A$ results in an allocation that is lexicographically smaller than $A$, a contradiction. 

  If an agent in $S_{<k}$ considered an item owned by an agent $S_{\ge k}$ heavy, we would have an alternating path starting with an $A$-edge from the latter agent to the former. Thus, no agent in $S_{< k}$ considers heavy any good that is assigned to an agent in $S_{\ge k}$, and all goods considered heavy by an agent in $S_{\ge k}$ but by no agent in $S_{< k}$ are assigned to the former agents.\end{proof}

We use $H$ to denote the lexmin allocation of the heavy goods.

% \begin{figure}
%  \centering \includegraphics[width=0.7\textwidth,clip, trim=0 500 0 70]{allocationLemma32.pdf}
% \end{figure}

\subsubsection{Assignment of Light Goods and Optimization of Small Bundles}\label{APC exists}

We begin by providing a characterization of an optimal allocation (under the restriction that any heavy good is  allocated to an agent that considers it heavy). We shall show how to efficiently compute such an optimal allocation in Section~\ref{Algorithmics}. We add the light goods greedily to the lexmin allocation of the heavy goods, i.e., we add the light goods one by one and always to a bundle of minimum weight. Let $x$ be the minimum weight of a bundle after the addition of the light goods. Then, only bundles of weight $x$, $x + \half$, and $x + 1$ contain light goods. The assignment of the light goods is somewhat arbitrary in the following sense: If a bundle of value $x + 1$ contains a light good, we may move this light good to a bundle of value $x$, thus interchanging the values of the bundles.

Let $k_0$ be minimal such that $k_0 s > x + 1$. Then, $(k_0 - 1) s \le x + 1$, and, hence, the bundles in $R'_{k_0}$ may contain light goods. They contain at most $\floor{s}$ light goods as $(k_0 - 1)s + \floor{s} + 1 > k_0s > x + 1$ and bundles of value greater than $x +1$ contain no light good; $\floor{s}$ goods are contained if $k_0 s = x + \sfrac{3}{2}$ and the value of the bundle is $x + 1$.

At this point, we have bundles of value $x$, $x + \half$, $x + 1$, $k_0s$, $(k_0 + 1)s$, \ldots\;. We use $B$ to denote the allocation after the greedy addition of the light goods. Let us call the bundles of value at most $x + 1$ the \emph{small bundles}, the bundles of value $k_0s$ and more the \emph{large bundles}, let us write $N_s$ for the owners of the small bundles, and let $\Gamma(N_s)$ be the goods owned by them. The agents in $R'_{k_0}$ belong to $N_s$, and the agents in $R_{k_0}$ do not belong to $N_s$. 

No allocation can assign more heavy goods to the agents in $N_s \cup R_{k_0}$ as $B$ does. Also $B$ allocates all light goods to the agents in $N_s$. So no allocation can assign more value to the agents in $N_s \cup R_{k_0}$ as $B$ does. If all bundles of $B$ in $N_s$ have value $x$, $B$ is clearly optimal. So we may assume that the average value of the bundles in $N_s$ lies strictly between $x$ and $x + 1$. 

\paragraph{The Optimal Allocation:} We optimize the allocation of the small bundles by reassigning the goods contained in them so as to maximize the number of bundles of value $x + \half$. More precisely, we consider an allocation of the goods in $\Gamma(N_s)$ to the agents in $N_s$ such that \label{characterization of the optimal allocation}
\begin{enumerate}
\item each heavy good in $\Gamma(N_s)$ is allocated to an agent in $N_s$ that considers it heavy,
\item all bundles have value $x$, $x + \half$, or $x + 1$, and 
\item the number of bundles of value $x + \half$ is maximum;
\item subject to the number of bundles of value $x + \half$ being maximum, the number of not-heavy-only bundles of value $x + 1$ is maximum.  
\end{enumerate}
We use $A$ to denote the allocation obtained after optimization of the small bundles. It consists of $H$ restricted to the bundles of value $k_0s$ and more plus the optimized allocation, call it $C$, to the agents in $N_s$. \emph{We will show that $A$ is optimal and can be computed in polynomial time.} (Theorem~\ref{main theorem: heavy as heavy} and Lemma~\ref{running time}.)\medskip

Let $n_0$, $\nhalf$, and $n_1$ be the number of bundles in $C$ of value $x$, $x + \half$, and $x + 1$, respectively, let $S$ be their total value, and let $n_s = \abs{N_s}$.  Then, $n_s = n_0 + \nhalf + n_1$, and $S = n_0x + \nhalf (x + \half) + n_1 (x + 1)$, and hence $S - n_sx= \nhalf/2  + n_1$. Thus, $n_1 = (S - n_sx) - \nhalf/2$, and $n_0 = n_s - \nhalf - n_1 = (n_s(x + 1) - S) - \nhalf/2$. The $\NSW$ of $C$ is, therefore,
\begin{equation}\label{large mhalf is good} x^{n_0} (x + \half)^{\nhalf} (x + 1)^{n_1} = x^{n_s(x+1) - S}(x + 1)^{S - n_sx} [(x + \half)^2/(x(x+1))]^{\nhalf/2}.
  \end{equation}
  Maximizing the number of bundles of value $x + \half$ in $C$ is, therefore, tantamount to maximizing the $\NSW$ of $C$. It is, however, not clear at this point that an optimal allocation must satisfy Condition 2.

An efficient algorithm for constructing $C$ will be discussed in Section~\ref{Algorithmics}. Let $O$ be an optimal allocation closest to $A$, i.e., with minimal cardinality of $A^H \oplus O^H$, and let $x_O$ be the minimum value of a bundle in $O$. We decompose $A^H \oplus O^H$ into alternating paths. A good has degree zero or two in $A^H \oplus O^H$; recall that we are working under the assumption that heavy goods must be allocated as heavy goods. If the degree of a good is two, one path goes through the good. For an agent $v$, we pair $A$- and $O$-edges in $A^H \oplus O^H$ as much as possible and leave the remaining $A$- or $O$-edges incident to $v$  as starting edges of paths. Let $\hdeg_O(i)$ and $\hdeg_A(i)$ be the number of heavy goods owned by $i$ in $O$ and $A$, respectively. Let $w_i^O$ and $w_i^A$ be the value of the bundles owned by $i$ in $O$ and $A$, respectively. The \emph{utility profile} of an allocation is the multiset of the values of its bundles. 

\begin{lemma}\label{x0} $x_O \le x + \sfrac{1}{2} \le k_0s - 1$. Bundles of value more than $k_0s$ contain no light good in $O$.  \end{lemma}
\begin{proof} Let $T_{k_0}$ be the heavy items owned by the agents in $S_{k_0}$ in $A$. In $A$,  the agents in $R'_{k_0}$ own $k_0 - 1$ heavy items in $T_{k_0}$ and the agents in $R_{k_0}$ own $k_0$ such items. Only agents in $S_{\ge k_0}$ consider goods in $T_{k_0}$ heavy (Lemma~\ref{Sk}), and, hence, $O$ must assign these goods to the agents in $S_{\ge k_0}$. It may assign some of these goods to agents in $S_{> k_0}$ and it may assign more than $k_0$ of these items to agents in $S_{k_0}$.  Let $R''_{k_0}$ be the $\abs{R'_{k_0}}$ agents in $S_{k_0}$ that own the least number of heavy goods in $O$; ties are broken arbitrarily. These agents together own at most as many heavy goods in $O$ as the agents in $R'_{k_0}$ in $A$.
Therefore, the total weight assigned by $O$ to the agents in $(N_s \setminus R'_{k_0}) \cup R''_{k_0}$ is at most the total weight assigned by $A$ to the agents in $N_s$. The average weight assigned by $A$ to an agent in $N_s$ is less than $x + 1$ and hence the same holds for the agents in $(N_s \setminus R'_{k_0}) \cup R''_{k_0}$ in $O$. 

If a bundle of value more than $k_0s$ contained a light good, moving the light good to a bundle of value $x_O$ would improve the $\NSW$ of $O$. \end{proof}

 We will first show that $O$ and $A$ agree on the heavy bundles. More precisely, 

\begin{lemma}\label{agree on heavy} $A^H$ and $O^H$ agree on $S_{\ge k_0}$. Bundles in $N_s$ contain at most $k_0 - 1$ heavy goods in $O$. \end{lemma}
\begin{proof}
We start with the first claim. Assume otherwise, and let $k \ge k_0$ be maximal such that $A^H$ and $O^H$ do not agree on the agents in $S_k$. Since $A^H$ and $O^H$ agree on all agents in $S_{>k}$, and since the heavy goods assigned by $A^H$ to the agents in $S_{\ge k}$ must be assigned to agents  in $S_{\ge k}$ in any allocation (Lemma~\ref{Sk}), the heavy goods assigned to the agents in $S_k$ by $A$ must also be assigned to them by $O$, i.e., $\cup_{i \in S_k} A_i^H \subseteq \cup_{i \in S_k} O_i^H$. Assume first, that $\hdeg_O(i) = \hdeg_A(i)$ for all $i \in S_k$. Then, $O_i^H = A_i^H$ for all $i \in S_k$, since $O^H$ is closest to $A^H$. Note, that we may change $O_i^H$ to $A_i^H$ for all $i \in S_k$ without affecting the NSW of $O$.  So, there must be an $i \in S_k$ that owns more heavy goods in $O$ than in $A$, i.e., $\hdeg_O(i) > \hdeg_A(i) \ge k-1$, and, hence, the heavy value of $i$ in $O$ is at least $k_0s$. Thus, $O_i$ contains no light good as, otherwise, the value of $O_i$ would be larger than $k_0s$, a contradiction to Lemma~\ref{x0}. Let $P$ be a maximal alternating path starting with an $O$-edge from $i$, and let $j$ be its other endpoint. Then, $P$ ends with an $A$-edge in $j$ and $\hdeg_O(j) < \hdeg_A(j)$. Let $\ell$ be the number of light goods owned by $j$. 

  Case $\hdeg_O(j) = \hdeg_O(i) - 1$: Augmenting the path to $O$ and also interchanging the light goods owned by $i$ and $j$ does not change the utility profile and, hence, the $\NSW$ of $O$. The reallocation moves $O$ closer to $A$, a contradiction.

  Case $\hdeg_O(j) \le \hdeg_O(i) - 2$: If $w_j^O < w_i^O$, we augment the path and  move $\min(\ell,\floor{s})$ light goods from $j$ to $i$. This moves $\OH$ closer to $\AH$ and does not decrease the $\NSW$ of $O$ (it improves it if $w_j^O \le w_i^O - 1$). If $w_j^O \ge w_i^O$, $j$ owns at least $2s$ light goods.  Since $w^O_i \ge k_0s$, and bundles of value larger than $k_0s$ contain no light good, we have $w_j^O = w_i^O = k_0s$ and $k_0s \le x_O + 1$. Thus, $k_0s = x_O + 1$. We move a light good from $j$ to a bundle of value $x_O$, augment the path, and move $\floor{s}$ light goods from $j$ to $i$. This improves the $\NSW$ of $O$, a contradiction. 

Case $\hdeg_O(j) \ge \hdeg_O(i)$: Then, $\hdeg_A(j) > \hdeg_O(j) \ge \hdeg_O(i) > \hdeg_A(i) \ge k - 1$. Thus,
$\hdeg_A(j) \ge \hdeg_A(i) + 2$, and there is an alternating path starting with an $A$-edge from $j$ to $i$. Thus, $A$ is not lexmin, a contradiction.
This completes the proof of the first claim.

For the second claim, assume that there is an agent $i \in N_s$ that owns $k_0$ or more heavy goods in $O$. Since $i$ owns at most $k_0 - 1$ heavy goods in $A$, we have $\hdeg_O(i) \ge k_0 > \hdeg_A(i)$. As above, we conclude that $O_i$ contains no light good and consider an alternating path $P$ starting with an $O$-edge in $i$. Let $j$ be the other endpoint. Then, $\hdeg_O(j) < \hdeg_A(j)$.

If $\hdeg_O(j) < hdeg_O(i)$, we argue as above. If $\hdeg_O(j) \ge \hdeg_O(i) \ge k_0$, then $\hdeg_A(j) > \hdeg_O(j) \ge \hdeg_O(i) > \hdeg_A(i)$, $P$ is an alternating path starting with an $A$-edge from $j$ to $i$, and $A$ is not lexmin.
This completes the proof of the second claim.
\end{proof}

\begin{lemma} $A$ and $O$ agree on the bundles $R_{k_0} \cup S_{> k_0}$, i.e., on the agents owning $k_0$ or more heavy goods in $A$. \end{lemma}
\begin{proof} By the preceeding Lemma, they agree on the heavy goods. Also, none of the bundles contains a light good in either $A$ or $O$.  \end{proof}

We still need to show that $A$ is the optimal allocation for the agents in $N_s$. \emph{Since we know at this point that $A$ and $O$ agree on the large bundles, we use $A$ and $O$ in the sequel for the bundles of $A$ and $O$ allocated to the agents in $N_s$}.

\begin{lemma}\label{less-more} If $A$ is suboptimal, there are more bundles of value $x + \sfrac{1}{2}$ in $O$ than in $A$ and fewer bundles of value $x + 1$.  \end{lemma}
\begin{proof} Let $m_d$ and $n_d$ be the number of bundles of value $x + d$ in $A$ and $O$, respectively. 
Only $m_0$, $\mhalf$, and $m_1$ are non-zero. There might be non-zero $n_d$ for $d < 0$ and $d > 1$. For the sake of a contradiction, assume $\nhalf \le \mhalf$. We massage $O$ by moving portions of $\half$ around; we maintain the number of bundles and their total value, but do not insist that $O$ stays a feasible allocation. Each move will strictly increase the $\NSW$, and we will end up with $A$. Thus $O = A$.

We observe first that the sums $\sum_{d \ge \sfrac{1}{2}} n_d$ and $\sum_{d \le \sfrac{1}{2}} n_d$ must both be non-zero, as, otherwise, the average weight of a bundle in $O$ would be at most $x$ or at least $x + 1$.

If $n_d = 0$ for $d < 0$ and $d > 1$, $\nhalf \le \mhalf$ implies $\NSW(O) \le \NSW(A)$ by equation (\ref{large mhalf is good}). If $n_d > 0$ for some $d < 0$ or $d > 1$, one of the following rules improves the $\NSW$ of $O$ and maintains $\nhalf \le \mhalf$.

If for some half-integers $e$ and $d$, $e + 1 \le d$ and $n_e > 0$ and $n_d > 0$, decrease $n_e$ and $n_d$ by one each and increase $n_{e + \sfrac{1}{2}}$ and $n_{d - \sfrac{1}{2}}$ by one each, except if this would increase $\nhalf$ above $\mhalf$.

If this rule is not applicable, we either have $n_d = 0$ for $d < 0$ and $d > 1$ and are done, or $\nhalf = 0 = \mhalf$, and $n_d = 0$ for $d \le -\sfrac{1}{2}$ or $d \ge \sfrac{3}{2}$, and $n_d$ is non-zero for some $d \not\in [0,1]$. If there are two values of $d$ outside $[0,1]$ with positive $n_d$, the values differ by $\sfrac{1}{2}$, as, otherwise, the rule applies.  Since the total value of the bundles in $A$ and $O$ must be the same\footnote{Recall that we are still operating under the assumption that heavy goods must be allocated to agents that consider them heavy.}, we have $
-\sum_{\ell \ge 1} \sfrac{\ell}{2} n_{-\sfrac{\ell}{2}} + 0n_0 + n_1 + \sum_{k \ge 3} \sfrac{k}{2} n_{\sfrac{k}{2}} = 0m_0 + m_1$, and either the sum indexed by $\ell$ or the sum indexed by $k$ is zero, and in any sum, only two consecutive terms can be non-zero. Thus, $n_{-\sfrac{\ell}{2}}$ and $n_{\sfrac{k}{2}}$ are even for odd $\ell$ and odd $k$.

In order to proceed, we need the following well-known inequality. For completeness, we include a proof as a footnote. 

\begin{claim}\label{useful ineq} Let $a_1$,\ldots,$a_s$, $b_1$, \ldots, $b_t$ be positive reals, let $\varepsilon_1$, \ldots, $\varepsilon_s$ and $\delta_1$,\ldots, $\delta_t$ be non-negative reals with $\sum_i \varepsilon_i = \sum_j \delta_j$ and $a_i + \varepsilon_i \le b_j - \delta_j$ for all $i$ and $j$. Then,
  \[    \prod_i(a_i + \varepsilon_i) \prod_j (b_j - \delta_j) \ge \prod_i a_i \prod_j b_j \]
  with strict inequality if some $\varepsilon_i$ or $\delta_j$ is non-zero.\footnote{Let $a$ and $b$ be positive reals with $a < b$, and let $\varepsilon \in [0,b-a]$. Then, $(a + \varepsilon)(b - \varepsilon) = ab + \varepsilon (b - a - \varepsilon) \ge ab$. The inequality is strict if $\varepsilon \in (0,b-a)$.

To prove the claim, we use this inequality repeatedly. If all $\varepsilon_i$ and $\delta_j$ are zero, the claim is obvious. Otherwise, after renumbering, we may assume $\varepsilon_1 > 0$ and $\delta_1 > 0$. We also assume $\varepsilon_1 \le \delta_1$, the other case being symmetric. Then,
\[ \frac{(a_1 + \varepsilon_1)(b_1 - \delta_1)}{a_1b_1} = \frac{(a_1 + \varepsilon_1)(b_1 - \varepsilon_1)}{a_1b_1} \cdot \frac{b_1 - \varepsilon_1 - (\delta_1 - \varepsilon_1)}{b_1 - \varepsilon_1} > \frac{b_1 - \varepsilon_1 - (\delta_1 - \varepsilon_1)}{b_1 - \varepsilon_1}.\]
We are now left with the claim for $a_1$ to $a_s$, $b_1 - \varepsilon_1$, $b_2$, to $b_t$, and $0$, $\varepsilon_2$, \ldots, $\varepsilon_s$, $\delta_1 - \varepsilon_1$, $\delta_2$, \ldots, $\delta_t$, and a simple induction completes the proof.}\end{claim}

If $n_{-\sfrac{\ell}{2}}$ is positive for some odd $\ell$, we decrease it by two, decrease $n_1$ by one (must be positive), increase $n_0$ by one, and increase $n_{-\sfrac{(\ell - 1)}{2}}$ by two. This leaves the number of bundles and their total value unchanged and increases $\NSW$ by an application of Claim~\ref{useful ineq} with $a_1 = a_2 = x - \sfrac{\ell}{2}$, $b_1 = x + 1$, $\varepsilon_1 = \varepsilon_2 = \sfrac{1}{2}$, and $\delta_1 = 1$. 

If $n_{-\sfrac{\ell}{2}}$ is positive for some even $\ell$, we decrease it by one, decrease $n_1$ by one (must be positive), increase $n_0$ by one, and increase $n_{-\sfrac{(\ell - 2)}{2}}$ by one. This leaves the number of bundles and their total value unchanged and increases $\NSW$ by an application of Claim~\ref{useful ineq} with $a_1 = x - \sfrac{\ell}{2}$, $b_1 = x + 1$, $\varepsilon_1 = 1$, and $\delta_1 = 1$.

Analogously, if $n_{\sfrac{k}{2}}$ is positive for some odd $k$, we decrease it by two, decrease $n_0$ by one (must be positive), increase $n_1$ by one, and increase $n_{\sfrac{(k - 1)}{2}}$ by two. This leaves the number of bundles and their total value unchanged and increases $\NSW$ by an application of Claim~\ref{useful ineq} with $a_1 = x$, $b_1 = b_2 = x + \sfrac{k}{2}$, $\varepsilon_1 = 1$, and $\delta_1 = \delta_2 = \sfrac{1}{2}$. 

If $n_{\sfrac{k}{2}}$ is positive for some even $k$, we decrease it by one, decrease $n_0$ by one (must be positive), increase $n_1$ by one, and increase $n_{\sfrac{(k - 2)}{2}}$ by one. This leaves the number of bundles and their total value unchanged and increases $\NSW$ by an application of Claim~\ref{useful ineq} with $a_1 = x$, $b_1 = x + \sfrac{k}{2}$, $\varepsilon_1 = 1$ and $\delta_1 = 1$.

Continuing in this way, we turn $O$ into an allocation with $n_{1/2} \le m_{1/2}$, $0n_0 + \sfrac{1}{2} n_{\sfrac{1}{2}} + n_1 = 0\cdot m_0 + \sfrac{1}{2} m_{\sfrac{1}{2}} + m_1$, and all $n_d$ with $d < 0$ or $d > 1$ being zero. Equation~\ref{large mhalf is good} on page~\pageref{large mhalf is good} then implies that $A$ is optimal, a contradiction.

 A similar argument shows that there are fewer bundles of value $x + 1$ in $O$ than in $A$. Assume, for the sake of a contradiction, $n_1 \ge m_1$. If $n_e > 0$ and $n_d > 0$ for some $e < 1$ and $d > 1$, we decrease both counts by one and increase $n_{e + \sfrac{1}{2}}$ and $n_{d - \sfrac{1}{2}}$ by one each. So, we may assume that either $n_e = 0$ for all $e < 1$ or $n_d = 0$ for all $d > 1$. The former is impossible, as the average weight of a bundle would then be at least $x + 1$. So, $\sum_{e < 1} n_e > 0$ and $\sum_{d > 1} n_d = 0$. Assume $n_e > 0$ for some $e < 0$. If $n_{\sfrac{1}{2}} > 0$, we decrease $n_e$ and $n_{\sfrac{1}{2}}$ by one each and increase $n_{e + \sfrac{1}{2}}$ and $n_0$ by one each. If $n_{\sfrac{1}{2}} = 0$, we have $n_1 > m_1$ since $\sum_d d n_d = \sum_d d m_d$, and we decrease $n_e$ and $n_1$ by one each and increase $n_{e + \sfrac{1}{2}}$ and $n_{\sfrac{1}{2}}$ by one each, improving the NSW and maintaining $n_1 \ge m_1$. In this way, all bundles have values in $\sset{x, x + \sfrac{1}{2}, x + 1}$. Since the number of bundles and their total value is fixed, having more bundles of value $x + \sfrac{1}{2}$ implies having fewer bundles of value $x$ and $x + 1$. 
  \end{proof}

\newcommand{\rhalf}{r_{\sfrac{1}{2}}}
\begin{lemma}\label{max number of heavy} Let $r_d$ be the maximum number of heavy goods that a bundle of value $x + d$ may contain. Then, $r_d = \max\set{r \in \NN_0}{x + d -sr \in \NN_0}$, and
  \begin{itemize}
  \item $r_0 = \rhalf - 1$ and $r_1 = \rhalf + 1$ if $x + 1$ is an integer multiple of $s$.
  \item $r_1 = r_0 = \rhalf + 1$ if $x + 1$ is not an integer multiple of $s$ and $x + 1 > (\rhalf + 1) s$.
  \item $r_1 = r_0 = \rhalf - 1$ if $x + 1$ is not an integer multiple of $s$ and $x + 1 < (\rhalf + 1) s$.
  \end{itemize}
    \end{lemma}
    
    \begin{proof} Observe first that $r_0 \le r_1 \le r_0 + 2$.  The first inequality is obvious, as adding a light to a bundle of value $x + d$ yields a bundle of value $x + d + 1$. The second inequality holds because we may remove two heavy goods from the heavier bundle and add $2s - 1$ light goods. Also, $r_0$ and $r_1$ have the same parity, since $x + d$ and $x + d+ 1$ are either both integral or both half-integral.  Finally $r_{\sfrac{1}{2}} - r_0 = \pm 1$ since the two numbers have different parity, and we can switch between the two values by exchanging a heavy good by either $\floor{s}$ or $\ceil{s}$ light goods.

 Let $x + \sfrac{1}{2} = \rhalf s + y$ with $y \in \NN_0$. Then, $x + 1 = \rhalf s + y + \sfrac{1}{2}$. If $y + \sfrac{1}{2} = s$, $r_1 = \rhalf +1$. Also, $r_0 < r_1$, and, hence, $r_0 = \rhalf-1$. If $y + \sfrac{1}{2} > s$, then $y - \sfrac{1}{2} \ge s$, and, therefore, $r_1 = r_0 = \rhalf + 1$. If $y + \sfrac{1}{2} < s$, then $x + 1 = (\rhalf - 1)s + (s + y + \sfrac{1}{2})$, and, hence, $x = (\rhalf - 1)s + (s + y - \sfrac{1}{2})$, and, therefore, $r_1 = r_0 = \rhalf - 1$. 
\end{proof}

The following Lemma is useful for arguing that a change of allocation is feasible. 

\begin{lemma}\label{global accounting} Let $S$ be the total value of all goods.\footnote{Recall that in this section, heavy goods must be allocated as heavy goods. So, $S$ is equal to the number of light goods plus $s$ times the number of heavy goods.} For each agent $i$, let $t_i$ be a target value for the agent such that $S = \sum_i t_i$, and let $D^H$ be an allocation of the heavy goods such that $t_i - \abs{D^H_i} s$ is a non-negative integer. Then, the light goods can be added to the bundles so as to give each bundle its target value, and all light goods are allocated. \end{lemma}
\begin{proof} Let $h_i = \abs{D^H_i}$. Then, the number of available light goods is $S - \sum_i h_i s$. We need $t_i - h_is$ light goods for agent $i$, and, hence, $\sum_i (t_i - h_i s) = S - \sum_i h_i s$ altogether. \end{proof}

For an agent $i$, $w^A_i$ is the value of $i$'s bundle in allocation $A$ for $i$. If the allocation is clear from the context, we drop the superscript.

\begin{lemma}\label{simple observation} Consider any allocation, agents $i$ and $j$, and assume $\hdeg(i) < \hdeg(j)$ and $w_i \not= w_j$. Let $\ell$ be the total number of light goods in the bundles of $i$ and $j$, and consider a good in the bundle of $j$ that both $i$ and $j$ consider heavy. Moving the heavy good from $j$ to $i$, and then adding the light goods greedily does not decrease the $\NSW$. \end{lemma}
\begin{proof} The weight difference before the reallocation is at least $\sfrac{1}{2}$. The difference in heavy weight changes by $2s$. The absolute difference in heavy weight decreases, except if $\hdeg(i) + 1 = \hdeg(j)$ when it is unchanged. Let $d$ be the difference in heavy weight after the reallocation of the heavy good. If $d > \ell$, the greedy allocation reduces the weight difference to $d - \ell$, and the difference before the reallocation was at least as large. If $d \le \ell$ and $d$ is half-integral, the greedy allocation reduces the weight difference to $\sfrac{1}{2}$. The  weight difference before was at least $\sfrac{1}{2}$. If $d$ is integral, the greedy allocation reduces the weight difference to either zero or one. The weight difference before was at least one since $w_i \not= w_j$ and the difference in heavy weight was integral before the reallocation. \end{proof}

\newcommand{\tw}{\mathit{tw}}  \newcommand{\SA}{S^A} \newcommand{\SO}{S^O}

We use $\SA_d$ and $\SO_d$ for the set of agents owning bundles of value $x + d$ in $A$ and $O$, respectively. 

  \begin{lemma}\label{xo}
    $x - \sfrac{1}{2} \le x_O \le x + \sfrac{1}{2}$. The bundles in $O$ have values in $\sset{x - \sfrac{1}{2}, x, x + \sfrac{1}{2}, x + 1}$. 
\end{lemma}
\begin{proof} $x_O \le x + \sfrac{1}{2}$ according to Lemma~\ref{x0}.
Assume $x_O \le x - 1$.  In $O$, all light goods are contained in bundles of value at most $x_O + 1$, and, hence, bundles of value larger than $x$ are heavy-only. Any bundle contains at most $k_0 - 1$ heavy goods (Lemma~\ref{agree on heavy}) and, hence, has heavy value at most $(k_0 - 1)s$.
Let $y = (k_0 - 1)s$. If $y \le x$, the average value of a bundle in $O$ is strictly less than $x$, since there is a bundle of value $x_O$ and all bundles have value at most $x$. But the average value of a bundle in $A$ is at least $x$, a contradiction. So $y \in \{x + \sfrac{1}{2}, x + 1\}$. In $O$, bundles of value $y$ are heavy-only. Let $U$ be the set of owners of the bundles of value $y$ in $O$. If their bundles in $A$ have value $y$ or more, the average value of a bundle in $A$ is larger than the average value in $O$, a contradiction. Note that bundles in $N_s \setminus U$ have value at least $x$ in $A$, have value at most $x$ in $O$, and there is a bundle of value $x_O$ in $O$.

  So, there is an agent $i \in U$ whose bundle in $A$ has value less than $y$. Then $\hdeg_A(i) \le k_0 - 2 < \hdeg_O(i)$. Consider an alternating path starting with an $O$-edge incident to $i$. It ends in a node $j$ with $\hdeg_O(j) < \hdeg_A(j) \le k_0 - 1$. Then, $\hdeg_O(j) <\hdeg_O(i)$, and the value of $O_j$ is at most $x$, and, hence, is less than $y$. Augmenting the path to $O$ effectively moves a heavy good from $i$ to $j$ and decreases the distance between $A^H$ and $\OH$. Adding the light goods in $O_i \cup O_j$ greedily will result in an allocation whose $\NSW$ is at least the one of $O$ (Lemma~\ref{simple observation}), a contradiction. 

We have now established $x - \sfrac{1}{2} \le x_O \le x+ \sfrac{1}{2}$. It remains to show that bundles in $O$ have value at most $x + 1$. Assume there is a bundle of value $x + \sfrac{3}{2}$ or more in $O$. By Lemma~\ref{agree on heavy}, the bundle
contains at most $k_0 - 1$ heavy goods, and, hence, its heavy value is at most $x + 1$. So, it contains a light good, and, hence, $x_O = x + \sfrac{1}{2}$, the bundle under consideration has value $x + \sfrac{3}{2}$, and the bundles in $O$ have values in $\{x + \sfrac{1}{2}, x + 1, x + \sfrac{3}{2}\}$. Any bundle of value $x + \sfrac{1}{2}$ can be turned into a bundle of value $x +\sfrac{3}{2}$ by moving a light good to it from a bundle of value $x + \sfrac{3}{2}$.

If $A$ is optimal, $O^H = A^H$, as $O$ is closest to $A$. Hence, $x_O = x_A = x$, and there are no bundles of value $x + \sfrac{3}{2}$ in $O$. If $A$ is sub-optimal, there are more bundles of value $x + \sfrac{1}{2}$ in $O$ than in $A$ (Lemma~\ref{less-more}). So, there is an $i$ in $(\SO_{\sfrac{1}{2}} \cup \SO_{\sfrac{3}{2}}) \cap (\SA_0 \cup \SA_1)$. Then, the parities of $\hdeg_A(i)$ and $\hdeg_O(i)$ are different.  We may assume $i \in \SO_{\sfrac{3}{2}}$.

Assume first that $\hdeg_O(i) > \hdeg_A(i)$. Then, there exists an alternating path starting in $i$ with an $O$-edge. The path ends in $j$ with $\hdeg_O(j) < \hdeg_A(j) \le k_0 - 1$. The heavy value of $O_j$ is at most  $(k_0 - 2)s$, which is at most $x + 1 - s$. Since the value of $O_j$ is at least $x_O$, $O_j$ contains at least $x_O - (x + 1 - s) = s - \sfrac{1}{2} = \floor{s}$ light goods. 
If $w_j^O = x + \sfrac{1}{2}$, we augment the path, move $\floor{s}$ light goods  from $j$ to $i$ and improve $O$, a contradiction. If $w_j^O = x + 1$, we augment the path and move $\floor{s}$ light goods from $j$ to $i$. This does not change the utility profile of $O$ and moves $O$ closer to $A$, a contradiction. If $w_j^O = x + \sfrac{3}{2}$, $j$ contains at least $\ceil{s}$ light goods. We augment the path, move $\floor{s}$ light goods from $j$ to $i$ and one light good from $j$ to a bundle of value $x_O$. This improves $O$, a contradiction. 

Assume next that $\hdeg_O(i) < \hdeg_A(i) \le k_0 - 1$. Then, there exists an alternating path starting in $i$ with an $A$-edge. The path ends in $j$ with $\hdeg_A(j) < \hdeg_O(j) \le k_0 - 1$. Since $O_i$ has value $x + \sfrac{3}{2}$ and $\hdeg_O(i) \le k_0 - 2$, $O_i$ contains at least $x + \sfrac{3}{2} - (x  + 1 - s) = \ceil{s}$ light goods. We augment the path to $O$ and remove $\ceil{s}$ light goods from $O_i$. So, the value of $O_i$ becomes $x + 1$. If $O_j$ has value $x + \sfrac{3}{2}$, we put $\floor{s}$ light goods on $j$ and one light good on a bundle of value $x_O$. If $O_j$ has value $x + 1$ or $x + \sfrac{1}{2}$, we put $\ceil{s}$ light goods on $j$. This either improves $O$ or does not change the utility profile of $O$ and moves $O$ closer to $A$, a contradiction. 
\end{proof}

\begin{lemma}\label{containsalight}\label{$x + 1$ contains a light}\label{heavy-only $x + 1$} 
Consider a heavy-only agent $i \in \SA_1$. 
Then, $O_i$ is heavy-only and has value $x + 1$. If all bundles in $\SA_1$ are heavy-only, $A$ is  optimal.
\end{lemma}
\begin{proof} Consider a heavy-only bundle $A_i$ of value $x + 1$. Then, $\hdeg_A(i) = k_0 - 1$, and $x + 1 = (k_0 - 1)s$. Assume for the sake of a contradiction that $O_i$ has either value less than $x +1$ or is not heavy-only. In either case, $\hdeg_A(i) > \hdeg_O(i)$, and, hence, $O_i$ contains at most $k_0 - 2$ heavy goods and thus at least $w_i^O - (x + 1 - s) = w_i^O - x - 1 + s$ light goods. If $i \in \SO_0 \cup \SO_1$, $O_i$ contains at most $k_0 - 3$ heavy goods as the parity of the number of heavy goods is the same as for $A_i$. 

  Consider an alternating path starting in $i$ with an $A$-edge, and let $j$ be the other end of the path. Then, $\hdeg_A(j) < \hdeg_O(j)$, and, hence, $A_j$ contains at most $k_0 - 2$ heavy goods; its heavy value is, therefore, at most $x + 1 - s$. Since the value of $A_j$ is at least $x$, $A_j$ contains  at least $\floor{s}$ light goods. At least $\ceil{s}$, if $A_j$ has value $x + 1$.

  If $w_j^A \in \sset{x,x+1}$, we augment the path to $A$, move $\floor{s}$ light goods from $j$ to $i$, and if $w_j^A = x + 1$, in addition, we move one light good from $j$ to a bundle of value $x$. In this way, we increase the number of bundles in $\Ahalf$ by two, a contradiction to property (3). 

  If $w_j^A = x + \sfrac{1}{2}$ and $A_j$ contains at least $\ceil{s}$ light goods\footnote{It will contain at least $2s + \floor{s}$ light goods.}, we augment the path to $A$ and move $\floor{s}$ light goods from $j$ to $i$. In this way, $i$ and $j$ interchange values, and $A_j$ becomes a bundle of value $x + 1$ containing a light good, a contradiction to property (3).

If $w_j^A = x + \sfrac{1}{2}$ and  $A_j$ contains exactly $\floor{s}$ light goods, it contains exactly $k_0 - 2$ heavy goods. Then, $O_j$ contains $k_0 - 1$ heavy goods and, hence, is a heavy-only bundle of value $x + 1$. If the value of $O_i$ is $x - \sfrac{1}{2}$, $O_i$ contains at least $s - \sfrac{3}{2}$ light goods. We augment the path to $O$ and move $s - \sfrac{3}{2}$ light goods from $i$ to $j$. This does not change the utility profile of $O$ and moves $O$ closer to $A$, a contradiction. If the value of $O_i$ is either $x$ or $x + \sfrac{1}{2}$, $O_i$ contains at least $\floor{s}$ light goods. We augment the path to $O$ and move $\floor{s}$ light goods from $i$ to $j$. This improves the $\NSW$ of $O$ if the value of $O_i$ is $x$. It does not change the utility profile of $O$ and moves $O$ closer to $A$ otherwise. If the value of $O_i$ is $x + 1$, $O_i$ contains at least $\ceil{s}$ light goods. We augment the path to $O$, move $\floor{s}$ light goods  from $i$ to $j$ and one light good from $i$ to a bundle of value $x$. This improves $O$. In either case, we have obtained a contradiction. 

We have now established that if $A_i$ is a heavy-only bundle of value $x + 1$, then $O_i$ is a heavy-only bundle of value $x + 1$. Therefore, if all bundles in $\SA_1$ are heavy-only, $\SA_1 \subseteq \SO_1$, and, hence, $A$ is optimal by Lemma~\ref{less-more}.
 \end{proof}

\newcommand{\Opm}{\SO_{\pm \sfrac{1}{2}}}

  We next decompose $A^H \oplus \OH$ into \emph{walk}s. Walks were previously used in studies of parity matchings~\cite{Akiyama-Kano}; see Section~\ref{Algorithmics}. For an edge in $A^H \setminus \OH$, we  say that the edge has type $A$. Similarly, edges in $\OH \setminus A^H$ have type $O$. We use $T \in \sset{A,O}$ for the generic type. Then, $\bar{T}$ is the other type. For $i \in \Ahalf \cap \Opm$, the heavy parity is the same in both allocations; we call the node \emph{unbalanced} if the cardinalities of $A^H_i$ and $\OH_i$ are different. For $i \in \Ahalf \cap \Opm$, and only for such $i$, we pair the edges in $A^H_i \oplus \OH_i$, first edges of different types and then the remaining either $A$- or $O$-edges (there being an even number of them). Goods have degree zero or two in $\OH \oplus \AH$; recall that heavy goods must be allocated as heavy goods in this section. For goods of degree two, we pair the incident $A$- and $O$- edge. In this way, we form walks. The endpoints of these walks are agents in $\SA_0 \cup \SA_1 \cup \SO_0 \cup \SO_1$. The interior nodes of the walks lie in $\Ahalf \cap \Opm$. We distinguish two kinds of interior nodes: through-nodes and hinges. In through-nodes, the two incident edges have different types. In hinges, the incident edges have the same type. Hinges are unbalanced. Depending on the type, the hinge is an $O$- or $A$-hinge. The types of the hinges along a walk alternate. Endpoints are also either $A$- or $O$- endpoints, depending on the type of the incident edge. The paths between hinges are alternating. Imagine we follow the walk from one endpoint to the other. If the first endpoint has type $T$, the first hinge will have type $\bar{T}$, the second hinge will have type $T$, the last hinge and second endpoint will have type $\bar{T}$ and $T$ or $T$ and $\bar{T}$, depending on whether the number of hinges is odd or even. 

Since there are more agents of value $x + \sfrac{1}{2}$ in $O$ than in $A$, there will be a walk $W$ that has  an endpoint $i \in \Ohalf \cap (\SA_0 \cup \SA_1)$. Let $j$ be the other endpoint. Then, $j \in \SA_0 \cup \SA_1 \cup \SO_0 \cup \SO_1$, as nodes in $\Ahalf \cap \Opm$ are interior nodes of walks. We will show: If $j \in \SA_0 \cup \SA_1$, $A$ does not satisfy its defining properties, and if $j \in (\SO_0 \cup \SO_1) \cap \Ahalf$, $O$ is not closest to $A$. 
Hinge nodes lie in $\SA_{\sfrac{1}{2}} \cap \SO_{\pm \sfrac{1}{2}}$. Hinge nodes actually lie in $\SO_{\sfrac{1}{2}}$, and $T$-hinges own at least $2s$ light goods in the allocation $\bar{T}$ for $T \in \{A,O\}$, as we show next.

\begin{lemma}\label{lights in hinges} All hinge nodes of $W$ belong to $\SO_{\sfrac{1}{2}}$, $A$-hinges own at least $2s$ light goods in $O$, and $O$-hinges own at least $2s$ light goods in $A$. \end{lemma}
\begin{proof} By  definition, hinge nodes are unbalanced and lie in
  $\SA_{\sfrac{1}{2}} \cap \SO_{\pm \sfrac{1}{2}}$. We will first show that all hinge nodes have the same value in $O$. If there is only one hinge node, this is obvious. Otherwise, consider two consecutive hinges and the alternating path $P$ connecting them. Assume that their values in $O$ differ by one, i.e, one has value $x + \sfrac{1}{2}$ and the other value $x - \sfrac{1}{2}$. We change $O^H$ to $O^H \oplus P$ and use Lemma~\ref{global accounting} to show that the modified heavy allocation can be extended to a full allocation. 
We define the target value of both endpoints as $x$ and the target value of all other agents as their value in $O$. Then, the sum of the values is unchanged. The $A$-endpoint, say $h$, of $P$ receives an additional heavy good; this is fine as $\hdeg_O(h) < r_{\sfrac{1}{2}}$, and, hence, $\hdeg_O(h) \le r_{\sfrac{1}{2}} - 2$, and, therefore, $\hdeg_O(h) + 1 \le r_{\sfrac{1}{2}} - 1 \le r_0$, according to Lemma~\ref{max number of heavy}. The $O$-endpoint loses a heavy good; this is fine since an $O$-endpoint must own at least one heavy good. The heavy degree of all other nodes stays unchanged. The modified allocation has larger $\NSW$ than $O$, a contradiction. We have now shown that consecutive hinge nodes, and hence all hinge nodes, have the same value in $O$.

It remains to show that the first hinge of the walk lies in $\SO_{\sfrac{1}{2}}$; call it $h$. Assume $h \in \SO_{-\sfrac{1}{2}}$. Let $i$ be the beginning of the walk. We distinguish cases according to whether $i$ is $A$-heavy or not. We define the target value of $i$ and $h$ as $x$ and the target value of all other nodes as their value in $O$. Then, again, the sum of the values does not change. If an allocation with the desired target values exists, it has larger $\NSW$ than $O$, a contradiction. 

If $i$ is $O$-heavy, $h$ is $A$-heavy. We augment $P$ to $O$. Then, $i$ loses a heavy good, which is fine as $\hdeg_O(i) > 0$, and $h$ gains a heavy good; this is fine as $h$ is an $A$-hinge, and, therefore, $\hdeg_O(i) \le \hdeg_A(i) - 2 \le r_{\sfrac{1}{2}} - 2$, and, hence,  $\hdeg_O(i) + 1 \le r_{\sfrac{1}{2}} - 1 \le r_0$. 

If $i$ is $A$-heavy, $h$ is $O$-heavy. We augment $P$ to $O$. Then; $h$ loses a heavy good; this is fine as $\hdeg_O(h) > 0$. The agent $i$ gains a heavy good; this is fine as $i$ is an $A$-endpoint, and, therefore $\hdeg_O(i) \le \hdeg_A(i) - 2 \le r_{\sfrac{1}{2}} - 2$, and, hence,  $\hdeg_O(i) + 1 \le r_{\sfrac{1}{2}} - 1 \le r_0$. 

We have now established that all hinges are unbalanced and belong to $\SA_{\sfrac{1}{2}} \cap \SO_{\sfrac{1}{2}}$. For a $T$-hinge $h$, $\hdeg_{\bar{T}}(v) \le \hdeg_T(v) - 2$. Thus, $h$ contains at least $2s$ light goods in $\bar{T}$.\end{proof}

  \begin{lemma} Any walk $W$ in the decomposition of $A^H \oplus \OH$ is a simple path, i.e., no agent occurs twice on $W$. \end{lemma}
  \begin{proof} Assume otherwise, say $W$ visits a vertex $i$ twice. Consider the path $W'$ between the two occurences of $i$. The first and the last edge of $W'$ have either different types or the same type. If they have the same type, $W'$ contains an odd number of hinges and, hence, together with $i$ forms a cycle with an even number of hinges. If the first edge and the last edge have different types, $W'$ contains an even number of hinges, maybe zero. If there is no hinge, $W'$ is simply an alternating cycle. If there is a hinge, we combine the alternating path from $i$ to the first hinge on $W'$ and the alternating path from the last hinge of $W'$ to $i$ into a single alternating path. So, in either case, $W'$ is a cycle with an even number of hinges that are connected by alternating paths. Half of the hinges are $A$-hinges and half are $O$-hinges. Let us pair the hinges so that each pair contains an $O$-hinge and an $A$-hinge. We augment $W'$ to $O$. Through-nodes gain and lose a heavy good, and, hence, the value of their bundle does not change. $O$-hinges lose two heavy edges and $A$-hinges gain two heavy edges.
Consider a pair $(h,\ell)$ of $A$- and $O$-hinges. An $A$-hinge owns at least $2s$ light goods in $O$. We give these goods to $\ell$. Then, the utility profile of $O$ does not change, and $O$ moves closer to $A$, a contradiction. \end{proof}

At this point, we have established the existence of a walk $W$ with an endpoint $i \in (\SA_0 \cup \SA_1) \cap \SO_{\sfrac{1}{2}}$. We may assume that there is a bundle of value $x + 1$ in $A$ containing a light good. Otherwise, $A$ is optimal (Lemma~\ref{containsalight}) and we are done. There is a bundle of value $x$ in $A$. Let $j$ be the other endpoint of the walk. Since $W$ is simple, $j \not= i$.
Assume $i$ is $T$-heavy, i.e., $\hdeg_T(i) > \hdeg_{\bar{T}}(i)$. The value of $A_i$ is $x$ or $x + 1$, and the value of $O_i$ is $x + \sfrac{1}{2}$. The types of the hinges alternate along $W$, the type of the first hinge is opposite to the type of $i$, and the type of the last hinge is opposite to the type of $j$. Each $A$-hinge holds $2s$ light goods in $O$, and each $O$-hinge holds $2s$ light goods in $A$ (Lemma~\ref{lights in hinges}). If the types of $i$ and $j$ differ, the number of hinges is even, and if the types are the same, the number of hinges is odd. 

We distinguish two cases: $j \in \SA_0 \cup \SA_1$ and $j \in (\SO_0 \cup \SO_1) \cap \SA_{\sfrac{1}{2}}$. In the former case, we show how to improve $A$, and, in the latter case, we will derive a contradiction to the assumption that $O$ is closest to $A$.

\paragraph{Case $j \in \SA_0 \cup \SA_1$:} We augment the walk to $A$. The heavy parity of $i$ and $j$ changes, and the heavy parity of all other nodes, including the hinges, does not change. Each $O$-hinge releases $2s$ light goods, and each $A$-hinge requires $2s$ light goods. An $A$-endpoint loses a heavy good and an $O$-endpoint gains a heavy good.

We define the target values of $i$ and $j$ as $\tw_i = \tw_j = x + \sfrac{1}{2}$. If $w^A_i = w^A_j = x$, let $k$ be an agent that owns a bundle of value $x + 1$ containing a light good and define $\tw_k = x$. If $w^A_i = w^A_j = x + 1$, let $k$ be an agent that owns a bundle of value $x$ and define $\tw_k = x + 1$. If $\sset{w^A_i,w^A_j} = \sset{x, x + 1}$, leave $k$ undefined. For all other agents, their target value is their value in $A$. Then, the sum of the target values is the same as the total value of the bundles in $A$.

We next define the number of heavy goods allocated to each agent in the new allocation. For an $A$-hinge $h$, the number decreases by two; since $\hdeg_A(h) \ge \hdeg_O(h) + 2 \ge 2$, this is fine. For an $O$-hinge $h$, the number increases by two; since $\hdeg_A(h) + 2 \le \hdeg_O(h) \le r_{\sfrac{1}{2}}$, this is fine. For an $O$-endpoint, the number increases by one; if $i$ is an $O$-endpoint, this is fine since $\hdeg_A(i) < \hdeg_O(i) \le r_{\sfrac{1}{2}}$; if $j$ is an $O$-endpoint and $j \in \Ohalf$, the same reasoning applies; if $j$ is an $O$-endpoint and $j \in \SO_0 \cup \SO_1$, $\hdeg_A(j) < \hdeg_O(j)$ implies $\hdeg_A(j) + 2 \le \hdeg_O(j) \le r_1$, and, hence, $\hdeg_A(j) + 1 \le r_{\sfrac{1}{2}}$. For an $A$-endpoint $h$, the number decreases by one; this is fine, since $\hdeg_A(h) > \hdeg_O(h) \ge 0$. Finally, the total number of allocated heavy goods stays the same. This holds, since the number of hinges is even if the endpoints have differnt types, and since there is an additional $T$-hinge if both endpoints have type $\bar{T}$. 

The total number of allocated heavy goods stays the same, the total value of the bundles stays the same, and each bundle contains at most the maximum feasible number of heavy goods. Therefore, the number of light goods is precisely the number of light goods needed to fill all bundles to their target value (Lemma~\ref{global accounting}). The transformation improves the $\NSW$ of $A$, and the values of all bundles in the new allocation lie in $\sset{x, x+ \sfrac{1}{2}, x + 1}$, a contradiction to the optimality of $A$ under these constraints.

\paragraph{Case $j \in (\SO_0 \cup \SO_1) \cap \SA_{\sfrac{1}{2}}$:} We augment the walk to $O$. For the intermediate nodes, the heavy parity does not change. For $i$ and $j$, the heavy parity changes. We define the target values as $\tw_i = w^O_j$ and $\tw_j = w^O_i$. For all other agents, the target value is their value in $O$. So the sum of the target values is the total value of the bundles in $O$, and the utility profile does not change.

If $i$ and $j$ have different types, the number of hinges is even. If $i$ and $j$ have type $T$, there is an extra hinge of type $\bar{T}$. If $T = A$, the endpoints gain one heavy good each, and the hinge loses two heavy goods. If $T = O$, the endpoints lose one heavy good each, and the hinge gains two heavy goods. So the sum of the heavy goods in all bundles does not change.

It remains to argue that no bundle contains too many heavy goods and that agents losing a heavy good actually owned one. An $O$-endpoint $h$ loses a heavy good; this is fine, since $\hdeg_O(h) > 0$ for an $O$-endpoint. An $A$-endpoint gains a heavy good. If $j$ is an $A$-endpoint, $\hdeg_O(j) < \hdeg_A(j) \le r_{\sfrac{1}{2}}$, and, hence, $\hdeg_O(j) + 1 \le r_{\sfrac{1}{2}}$; this is fine, since $\tw_j = w^O_i = x + \sfrac{1}{2}$. If $i$ is an $A$-endpoint and $A_i$ has value $x$, $\hdeg_O(i) < \hdeg_A(i) \le r_0$, and, hence, $\hdeg_O(i) + 1 \le r_0$. This is fine, since $\tw_i = w_j^O \in \sset{x, x + 1}$. If $i$ is an $A$-endpoint and has value $x + 1$, $A_i$ cannot be heavy-only since then $O_i$ would also have value $x + 1$ and be heavy-only according to Lemma~\ref{heavy-only $x + 1$}. Hence, $\hdeg_O(i) \le \hdeg_A(i) - 2 \le r_1 - 2$. So, $\hdeg_O (i) + 1 \le r_1 -1 \le r_O$, and we are fine. 

Thus, the utility profile of $O$ does not change, and $O$ moves closer to $A$, a contradiction.

We have now established: 

\begin{theorem}\label{main theorem: heavy as heavy} $A$ is optimal.\end{theorem}

\subsubsection{Algorithm for Optimizing Small Bundles}\label{Algorithmics}

We now show how to efficiently compute an optimal allocation (under the assumption that any heavy good is allocated to an agent who considers it heavy), using the characterization found in Section~\ref{APC exists}. Let $x$ be the minimum value of a bundle after the greedy addition of the light goods. Then, only
bundles of value $x$, $x + \sfrac{1}{2}$, and $x + 1$ can contain small goods. Let $N_s$ be the
owners of the bundles of value at most $x + 1$, and let $\Gamma(N_s)$ be the goods assigned to $N_s$
in $A$. We need to construct the allocation $C$ of the goods in $\Gamma(N_s)$ to $N_s$ which
maximizes the number of bundles of value $x + \sfrac{1}{2}$.  All alternating paths in this section
are with respect to $A$ and $\bar{A}$.

\begin{lemma}\label{easy case} Let $A$ be an allocation with all values in
  $\{x, x + \sfrac{1}{2}, x + 1\}$. If there is either no bundle of value $x$ or no bundle of value
  $x + 1$, $A$ is optimal. \end{lemma}
\begin{proof} In either case, the number of bundles of each type is determined by the total value of
  the bundles and the number of bundles. \end{proof}

\begin{lemma}\label{create $x + 1$ containing a light} Assume all bundles in $\SA_1$ are
  heavy-only. If $A$ is sub-optimal, there is an alternating path starting with an $A$-edge from an
  agent in $\SA_1$ and ending in either $\SA_0$ or in an agent in $\SA_{\sfrac{1}{2}}$ that owns
  more than $\floor{s}$ light goods. \end{lemma}
\begin{proof} Since $x + 1$ is a multiple of $s$, bundles of value $x$ contain at least $2s - 1$
  light goods, and bundles of value $x + \sfrac{1}{2}$ contain at least $\floor{s}$ light
  goods. Consider alternating paths starting with an $A$-edge at agents in $\SA_1$ and assume that
  all such paths either lead to an agent in $\SA_1$ or to an agent in $\SA_{\sfrac{1}{2}}$ that owns
  exactly $\floor{s}$ light goods. Let $U$ be this set of agents. All heavy goods owned by agents in
  $U$ are light for agents not in $U$, as otherwise, more agents would belong to $U$. Also exactly
  $\floor{s}$ light goods are allocated to every agent in $U \cap \SA_{\sfrac{1}{2}}$. The bundles
  owned by the agents in $\bar{U}$ have value $x$ and $x + \sfrac{1}{2}$, and the maximal number of
  heavy and light goods are assigned to the agents in $\bar{U}$. So, the allocation is
  optimal. \end{proof}

For $d \in \sset{0,\sfrac{1}{2},1}$, let $N_d = \set{\ell}{\text{$\ell \le r_d$ and $\ell$ and $r_d$
have the same parity}}$ be the set of permissible numbers of heavy goods in bundles of value
$x + d$.

\begin{lemma}\label{parity matching} Let $A$ be an allocation with all values in
  $\{x, x + \sfrac{1}{2}, x + 1\}$. $A$ is sub-optimal if and only if there is an allocation $C^H$
  of the heavy goods in $A$ and a pair of agents $i$ and $j$ in $\SA_0 \cup \SA_1$ such that in $C$
  \begin{itemize}
  \item all agents in $\SA_{\sfrac{1}{2}} \cup \{i,j\}$ own a number of heavy goods in $N_{1/2}$, and
    for each of the agents in $\SA_0 \cup \SA_1 \setminus \{i,j\}$, the number of owned heavy goods
    is in the same $N$-set as in $A$; and,
  \item if $i$ and $j$ own bundles of value $x$ in $A$, there must be a bundle of value $x + 1$ in
    $A$ containing a light good, and if $i$ and $j$ own bundles of value $x + 1$ in $A$, there must
    be a bundle of value $x$ in $A$.
  \end{itemize}
\end{lemma}
\begin{proof} If $A$ is sub-optimal, there is an improving walk $W$.  Let $i$ and $j$ be the
  endpoints of the walk. Augmenting the walk and moving the light goods around, as described in
  Section~\ref{APC exists},
  \begin{itemize}
  \item adds $i$ and $j$ to $\SA_{1/2}$,
  \item reduces the value of a bundle of value $x + 1$ containing a light good to $x$ if $A_i$ and
    $A_j$ have value $x$ and increases the value of a bundle of value $x$ to $x + 1$ if $A_i$ and
    $A_j$ have value $x + 1$, and
  \item leaves the value of all other bundles unchanged.
  \end{itemize}
  Thus, in the new allocation, the number of heavy goods owned by $i$ and $j$ lies in
  $N_{\sfrac{1}{2}}$. For all other agents, the number of owned heavy goods stays in the same
  $N$-set. This proves the only-if direction.

  We turn to the if-direction. Assume that there is an allocation $C^H$ of the heavy goods in which
  for two additional agents $i$ and $j$ the number of owned heavy goods lies in $N_{1/2}$, and for
  all other agents, the number of owned heavy goods stays in the same $N$-set. We will show how to
  allocate the light goods such that $C^H$ becomes an allocation $C$ in which all bundles have
  value in $\{x, x + \sfrac{1}{2}, x + 1\}$ and
  $C_{\sfrac{1}{2}} = \SA_{\sfrac{1}{2}} \cup \{i,j\}$. Then, the $\NSW$ of $C$ is higher than the
  one of $A$.

  We next define the values of the bundles in $C$ and, in this way, fix the number of light goods that
  are required for each bundle. For $i$ and $j$, we define $\tw^B_i = \tw^B_j = x +
  \sfrac{1}{2}$. If $A_i$ and $A_j$ both have value $x$, let $k$ be an agent that owns a bundle of
  value $x +1$ containing a light good, and define $w_k^C = x$. If $A_i$ and $A_j$ have both value
  $x + 1$, let $k$ be an agent that owns a bundle of value $x$, and define $w_k^C = x + 1$.  Then,
  $w_i^A + w_j^A + w_k^A = w_i^C + w_j^C + w_k^C$ in both cases. If one of $A_i$ and $A_j$ has value
  $x$ and the other one has value $x + 1$, let $k$ be undefined. Then,
  $w_i^A + w_j^A = w_i^C + w_j^C$. For all $\ell$ different from $i$, $j$, and $k$, let
  $w_\ell^A = w_\ell^C$. Then, the total value of the bundles in $A$ and $C$ is the same. The total
  heavy value is also the same. Also, for each agent, the heavy value is at most the target
  value. So, the total difference between the target values and the heavy values is exactly the same
  in $C$ as in $A$. Hence, $C$ exists. \end{proof}

In a \emph{parity matching problem} in a bipartite graph, a maximum degree is specified for every
vertex. The question is whether there is a subset of the edges such that the degree of every vertex
is at most its maximum and has the same parity as its maximum. Parity matching problems can be
solved in polynomial time~\cite{CornuejolsFactors,Akiyama-Kano}.

\begin{lemma}\label{running time} Under the restriction that heavy goods must be
  allocated as heavy goods, an optimal allocation for $N_s$ can be computed in polynomial time, and
  the $\NSW$-allocation for half-integer instances can be computed in polynomial time.
\end{lemma}
\begin{proof} Let $A$ be any allocation for $N_s$. All bundles have values in
  $\sset{x, x + \sfrac{1}{2}, x + 1}$. If there are no bundles of value $x$ or no bundles of value
  $x + 1$, $A$ is optimal (Lemma~\ref{easy case}).

  If all bundles of value $x + 1$ are heavy-only, we search for an alternating path starting with an
  $A$-edge from an agent in $\SA_1$ and ending either in an agent in $\SA_0$ or in an agent in
  $\Ahalf$ that owns more than $\floor{s}$ light goods. If such a path exists, we either improve $A$ or
  create a bundle of value $x + 1$ containing a light good. We repeat until we create a bundle of
  value $x +1$ containing a light good. If we do not succeed, $A$ is optimal (Lemma~\ref{create
  $x + 1$ containing a light}).

  So, assume now that we have a bundle of value $x$ and a bundle of value $x + 1$ containing a
  light. We then check whether $A$ can be improved according to Lemma~\ref{parity matching}.  In
  order to check for the existence of the allocation $C^H$, we set up the following parity matching
  problem for every pair $i$ and $j$ of agents.
  \begin{itemize}
  \item For goods, the degree in the matching must be equal to 1.
  \item For all agents in $\SA_{\sfrac{1}{2}} \cup \{i,j\}$, the degree must be in
    $N_{\sfrac{1}{2}}$.
  \item If $A_i$ and $A_j$ have value $x$, let $A_k$ be any bundle of value $x + 1$ containing a
    light good. If $A_i$ and $A_j$ have value $x + 1$, let $A_k$ be a bundle of value $x$. The
    degree of $k$ must be in $N_0$.\footnote{It would be incorrect to require that the degree of $k$
    must be in $N_1$ since we want to allocate a light good to $k$.}
  \item For an $a \in \SA_0\setminus \{i,j,k\}$, the degree must be in $N_0$, and for an
    $a \in \SA_1 \setminus \{i,j,k\}$, the degree must be in $N_1$.
  \end{itemize}
  If $C^H$ exists for some pair $i$ and $j$, we improve the allocation. If $C^H$ does not exist for
  any pair $i$ and $j$, $A$ is optimal.

  Each improvement increases the size of $\SA_{\sfrac{1}{2}}$ by two, and, hence, there can be at
  most $n/2$ improvements. In order to check for an improvement, we need to solve $n^2$ parity
  matching problems. Parity matching problems can be solved in polynomial time.
\end{proof}

% We conjecture that this time bound can be improved.

\subsection{General Case: Heavy Goods can be Allocated as Light}\label{heavy as light}

Finally, we show our main theorem.

\HalfIntegerResult*

Let $O$ be an allocation maximizing NSW, and, for $t \ge 0$, let $A^t$ be a best allocation in which $t$ heavy goods are converted to light, i.e., the valuation functions are changed such that these $t$ goods are light for all agents. We will compute $\max_t \NSW(A^t)$.

\begin{lemma} $\NSW(O) = \max_t \NSW(A^t)$. \end{lemma} 
\begin{proof} Clearly, $\NSW(O) \ge \NSW(A^t)$ for all $t$. Consider the set of heavy goods that are allocated as light goods in $O$, let $t$ be their number, and change the valuation functions such that all agents consider these goods light. Then $\NSW(O) \le \NSW(A^t)$. 
\end{proof}

We follow the approach taken in the integral case in Section \ref{sec:integer}. We determine allocations $A^t$ for $t = 0, 1, 2, \ldots$\; ; we determined $A^0$ in the previous section (using the notation $A$ there). % Figure~\ref{structure of $A$} shows the structure of $A^0$. 
In $A^t$, $t$ goods are converted. Which $t$ goods? We will next derive properties of the optimal set of converted goods. 
%Let $x$ be the minimum value of any bundle and let $k_0$ be minimal such that $k_0 s > x + 1$. The \emph{core} of $A^t$ consists of all agents $a$ owning at most $k_0 - 1$ heavy goods and to which one cannot push an additional heavy item from an agent owning at least $k_0$ heavy items, i.e., there is no $A$-$\bar{A}$ alternating path from an agent owning at least $k_0$ heavy items to $a$. 

For a set $G$ of goods, let $C(G)$ be the allocation $A^0$ interpreted with the following modified valuation function: The goods in $G$ are light for all agents, and for any agent $i$, $i$ owns the heavy goods $A_i^H \setminus G$ and $\abs{A_i^H \cap G}$ light goods in addition to the light goods already owned by it.\footnote{We write $A_i$ instead of $A^0_i$ for the bundles of $A^0$.} Let $B(G)$ be an optimal allocation for the modified valuation function closest to $C(G)$, i.e., with minimal $\abs{C^H(G) \oplus B^H(G)}$. We choose $G$ such that
\begin{itemize}
\item $\NSW(B(G))$ is maximum, and
\item among the sets $G$ that maximize $\NSW(B(G))$, $\abs{G}$ is minimum, and
\item among the minimum cardinality sets $G$ that maximize $\NSW(B(G))$, $\abs{C^H(G) \oplus B^H(G)}$ is smallest.
\end{itemize}
For simplicity, let us write $B$ and $C$ instead of $B(G)$ and $C(G)$ for this choice of $G$. We also write $x_j$ for the value of $j$'s bundle in $C$ and $x'_j$ for the value of $j$'s bundle in $B$. We use $S$ to denote the set of agents $i$ with $A_i^H \cap G \not= \emptyset$. In Lemmas~\ref{conversion lemma}, \ref{further properties of G}, and~\ref{heaviest contributes to G}, we derive properties of $G$. 

\begin{lemma}\label{conversion lemma} Let $G$ and $S$ be as defined above, and let $x$ be the minimum value of any bundle in $A^0$. Then, $x'_\ell \le s + \min_{i \in S} x'_i \le \min_{i \in S} x_i$ for all agents $\ell$. For $i \in S$: $B_i$ contains no light good, $B_i^H \subseteq A_i^H \setminus G$, $x - 1  \le x'_i \le x_i - \abs{A_i^H \cap G} \le x_i - s$, and $x_i > x + 1$. 
%Finally, let $x_B$ be the minimum value of any bundle in $B$. Then $x_B \le \min_{i \in S} x'_i/s$. 
\end{lemma}
\begin{proof} If $G$ is empty, the Lemma obviously holds. So assume $G \not= \emptyset$, and, hence, $S \not= \emptyset$. Let $j = \argmin_{i \in S} x'_i$, and assume there is an $\ell$ such that $x_\ell' > x'_j + s$. Then, $\ell$'s bundle in $B$ is heavy-only, and, hence, $\ell$ owns at least two heavy items more than $j$ in $B$. We return one of the converted items to $j$, reconvert it to a heavy item, and replace it with a converted one from $\ell$. This changes $x'_j$ to $x'_j + s$ and $x'_\ell$ to $x'_\ell - s$ and, hence, improves NSW.\footnote{Note that $(x'_j + s)(x'_\ell - s) - x'_j x'_\ell = s(x'_\ell - x'_j - s) >0$.}

For the other claims, consider any $i \in S$, and let $g \in A_i^H \cap G$. Let $j$ be the agent to which $g$ is allocated in $B$; $j = i$ is possible.

Assume first that $B_i$ contains a light good, say $g'$.  We interchange $g$ and $g'$, i.e., we allocate $g'$ to $j$ and $g$ to $i$. This does not change the $\NSW$. We now reconvert $g$ back to a heavy good and improve the $\NSW$, a contradiction to the optimality of $B$. So, $B_i$ contains no light good.

  Assume next that $B_i^H \setminus C_i^H$ is non-empty. With respect to $B^H \oplus C^H$, goods have degree zero or two. Also, $B^H \oplus C^H$ contains no cycles, as augmenting a cycle to $B^H$ would decrease the distance between $C^H$ and $B^H$. Thus, $B^H \oplus C^H$ is a collection of $B$-$C$-alternating paths. One of these paths, call it $P$, starts in $i$ with a $B$-edge and ends with a $C$-edge $(g',h)$ at some agent $h$. We augment $P$ to $B$, re-allocate $g$ to $i$ as a heavy good, and convert $g'$ to a light good and give it to $j$; recall that $j$ is the owner of $g$ in $B$. The values of all bundles stay unchanged. We obtain an allocation $D$ with the same NSW as $B$, $G' = (G \setminus g) \cup g'$ as the set of converted goods, and $\abs{D^H \oplus C^H(G')} < \abs{B^H\oplus C^H}$, a contradiction to the choice of $G$. So, $B_i^H \subseteq C_i^H = A_i^H \setminus G$, and, hence, $x'_i \le x_i - \abs{A_i^H \cap G} \le x_i - s$, since $B_i$ contains no light good.

  From the preceding and the first paragraph of the proof, we obtain $x'_\ell \le s + \min_{i \in S} x'_i \le s + \min_{i \in S} x_i - s = \min_{i \in S} x_i$ for any agent $\ell$. 

  We next prove $x'_i \ge x - 1$ for all $i \in S$. If $S$ comprises all agents, $x'_i \le x_i -s$ for all agents, and, hence, $B$ is not optimal. So, assume that $S$ does not comprise all agents, and there is an agent $i^* \in S$ with $x'_{i^*} \le x - \sfrac{3}{2}$. Consider the following allocation $D$. Starting with $C(G)$, we move for all $i \in S$ the light goods in $C_i$ to agents outside $S$ and the heavy goods in $C_i^H \setminus B_i^H$ to their owners in $B$. At this point, the allocation agrees with $B$ for all $i \in S$. We next turn to the agents in $\bar{S}$. They own the items $\cup_{i \in \bar{S}} A_i^0$ plus maybe some heavy items previously owned by agents in $S$ plus at least one light item that was previously owned as a heavy item by an agent in $S$. Consider any optimal allocation of these items to the agents in $\bar{S}$, e.g., the one in $B$. We claim that every agent owns at least a bundle of value $x$ in this allocation. This can be seen as follows. We start with the allocation $A^0$ restricted to the agents in $\bar{S}$. Every bundle has value at least $x$. Then we add the additional heavy items and then we add the light items greedily. An alternative view of the greedy addition of the light items is as follows: As long as there is a bundle containing a light item and having a value that is more than one larger than another bundle, move the light item to the other bundle. When we start this process with bundles all having value at least $x$, we will never create a bundle of value less than $x$. Optimizing the allocation will also not create a bundle of value less than $x$. Note that the optimal allocation is characterized in the four bullets on page~\pageref{characterization of the optimal allocation}. One of the bundles owned by an agent in $\bar{S}$ contains a light item. We move it to $i^*$ and improve the NSW.
  
  We finally show $x_i > x + 1$ for all $i \in S$. For $s \ge \sfrac{5}{2}$, we have $x_i \ge x_i' + s \ge x - 1 + s > x+1$. So, assume $s = \sfrac{3}{2}$: If $x_i = x$, we have $x'_i \le x - \sfrac{3}{2}$, and if $i$'s bundle in $A^0$ contains a light item, we have $x'_i \le x + 1 - 1 - s \le x - \sfrac{3}{2}$, a contradiction to $x'_i \ge x - 1$. We are left with the $i$ such that $x_i \in \sset{x + \sfrac{1}{2}, x + 1}$, and $i$'s bundle in $A^0$ is heavy-only. Let $S' \subseteq S$ be the set of such $i$, and assume $S' \not= \emptyset$. For $i \in S'$, we have $x_i' \le x_i - s < x$. So, any bundle in $B$ of value $x + 1$ or more is heavy-only. Let $i_0 \in S$, and let $g\in A_{i_0} \cap G$. For $i \in S'$, $A_i$ contains exactly one item in $G$ and $B_i^H = C_i^H$, as $x'_i < x - 1$ otherwise. We cannot have that $S'$ comprises all agents, as there is a light good, and bundles $B_i$, $i \in S'$, are heavy-only.

We next show that $B$ is sub-optimal. We start with $C$ and move any light item in a bundle in $S'$ to a bundle outside $S'$. We then have $C_i = B_i$ for $i \in S'$. We next re-optimize the bundles outside $S'$ to obtain $B$. Since the NSW of $C$ is worse than the NSW of $A^0$, and since in $B$ all bundles of value $x + 1$ or more are heavy-only, the re-optmization must create an agent $h$ that owns a bundle of value $x + \sfrac{1}{2}$ containing a light good. Let $j$ be the agent to which $g$ is allocated in $B$. Then $x'_j \le x + \sfrac{1}{2}$. We move a light good from $h$ to $j$, and we move $g$ from $j$ to $i_0$ and reconvert it to a heavy good. The multiplicative change in NSW is at least 
  \[      \frac{x - \sfrac{1}{2}}{x + \sfrac{1}{2}}\cdot \frac{x + 1}{x - \sfrac{1}{2}} = \frac{x + 1}{x + \sfrac{1}{2}}, \]
  and, hence, the change improves NSW, a contradiction to the optimality of $B$. 
  \end{proof}

  \begin{corollary}\label{only small implies optimality} Let $x$ be the minimum value of any bundle in $A^0$. If $A^0$ contains no bundle of value more than $x + 1$, $A^0$ is optimal. \end{corollary}
  \begin{proof} If $A^0$ is not optimal, $S \not= \emptyset$. Finally, $x_i > x + 1$ for $i \in S$ by Lemma~\ref{conversion lemma}.\end{proof}

  %In $B$ we have bundles of value $\sset{x_B, x_B + \sfrac{1}{2}, x_B + 1}$ and then heavy-only bundles of value $k_B s$, $(k_B + 1)s$, \ldots, $m_B s$, where $k_B = \argmin_k ks > x_B + 1$ and $m_B s$ is the maximum value of any bundle in $B$. If there are no bundles of value larger than $x + 1$ in $B$, $k_B$ and $m_B$ do not exist. Agents in $S$ own bundles of value $m_B s $ or $(m_B - 1)s$ and $x_B \le m_B - 1$. 

Recall that $k_0 = \argmin_k ks > x+1$, and that for $k \ge k_0$, $R_k$ denotes the agents owning exactly $k$ heavy items in $A^0$ and to which no heavy item can be pushed from $S_{> k}$, and that $R'_k$ denotes the agents owning exactly $k - 1$ heavy items in $A^0$ to which a heavy item can be pushed from $R_k$. 

\begin{lemma}\label{further properties of G} Let $k_1$ be such that $(k_1 - 1)s = \min_{i \in S} x'_i$. Then, $k_1 \ge k_0$, and $R_{\ge k_1 + 1} \subseteq S \subseteq R_{\ge k_1}$. Also, $B_i = C_i^H \setminus G$ for $i \in S$. \end{lemma} 
\begin{proof} Since no heavy item in a bundle of value $x +1$ or less is converted, we have $k_1 \ge k_0$. Furthermore, by Lemma~\ref{conversion lemma}, $x'_i \in \sset{(k_1 - 1)s, k_1s}$ for all $i \in S$. So, $S \subseteq R_{\ge k_1}$. Then, $G$ contains $\ell - k_1$ heavy items in any bundle in $R_\ell$ for $\ell > k_1$, and
the remaining $r = \abs{G} - \sum_{\ell > k_1} (\ell - k_1) \abs{R_\ell}$ heavy items in $R_{\ge k_1}$. After the conversion, all bundles in $R_{\ge k_1}$ contain at most $k_1$ heavy items, and at least $r$ of them contain at most $k_1 - 1$ heavy items. 

We next show $B_i^H = C_i^H$ for $i \in S$. Assume, for the sake of a contradiction, that there is an $i \in S$ with $C_i^H \setminus B_i^H \not= \emptyset$. Since $(k-1)s \le x'_i < x_i - \abs{A_i \cap G}s = ks$, we have $x'_i = (k - 1)s$. As above, we conclude that $B^H \oplus C^H$ contains no cycle and, hence, decomposes into paths. Then, there is a $B^H \oplus C^H$ alternating path $P$ connecting $i$ and $h$, starting with a $C$-edge in $i$ and ending with a $B$-edge in $h$. Furthermore, the heavy degree of $h$ in $B$ is larger than its heavy degree in $C$. Since $B_h^H \not\subseteq C_h^H$, we have $h \not\in S$, and, hence, $h \in R_{\le k_1}$. Write $R_{\le k_1} = (R_{< k_1} \setminus R'_{k_1}) \cup R'_{k_1} \cup R_{k_1}$. 

We cannot have $h \in R_{< k_1} \setminus R'_{k_1}$, as agents in $R_{< k_1} \setminus R'_{k_1}$ consider goods owned by agents in $R'_{k_1} \cup R_{\ge k_1}$ light. More precisely, trace $P$ from $i$ until an agent $h' \in R_{< k_1} \setminus R'_{k_1}$ is reached. The good $g'$ preceding $h'$ is owned by an agent in $R'_{\ge k_1} \cup R_{\ge k_1}$ in $C$ and, hence, is considered light by $h'$. 

We cannot have $h \in R_{k_1}$, as then $x_h' \ge (k_1 + 1)s$, a contradiction. 

So, assume $h \in R'_{k_1}$. Then, $x'_h \ge ks$. Consider the allocation $D$ obtained by augmenting $P$ to $B$. The heavy degree of $i$ increases to $k$ and the heavy degree of $h$ decreases to at most $k - 1$. Thus, the NSW does not decrease, and $D$ is closer to $C$ than $B$, a contradiction to the choice of $B$. 

Since $B_i$ is heavy-only, $B_i^H = C_i^H$ implies $B_i = C_i^H$.  \end{proof}

\begin{lemma}\label{heaviest contributes to G} Either $G = \emptyset$ or for every $i \in N$ such that $v_i(A_i)$ is maximal and for every $g \in A_i$ there exists an optimal $G$ such that $g \in G$. \end{lemma}   

\begin{proof} We may assume $G \ne \emptyset$. Let $i \in N$ be such that $v_i(A_i)$ is maximal and $g \in A_i$. We may assume $g \notin G$. 

Assume first that $x'_i < x_i$ and $i$ is incident to a $C$-edge in $C^H \oplus B^H$. Let $P$ be a maximal alternating path starting at $i$ with an edge in $C^H \setminus B^H$. Since every good has even degree in $C^H \oplus B^H$, $P$ ends at an agent $k$ with a $B$-edge. Then, $k \not\in S$ by Lemma~\ref{conversion lemma} and, by Claim 5 from Section 4.2.2, $A_k$ contains at most one heavy good less than $A_i$, so $x'_k \ge x_i > x'_i$. We augment $P$ to $B$. If $x'_k \ge x'_i + s$, the NSW of $B$ does not decrease. Otherwise, $i$ has $\ceil{x'_i - (x_i-s)}$ light goods that we give to $k$. Now, $i$ has value $x'_i + s - \ceil{s + x'_i - x_i} = x'_i + (x_i - x'_i -\delta)$, and $k$ has value $x'_k - s + \ceil{s + x'_i - x_i} = x'_k - (x_i - x'_i - \delta)$, where $\delta \in \sset{0,\sfrac{1}{2}}$. Since the sum of their utilities does not change and $0 \le x_i - x'_i - \delta \le x_k' - x'_i$, the NSW of $B$ does not decrease. But $B$ moves closer to $C$, a contradiction to the choice of $B$. So, we have either $x'_i \ge x_i$ or $C^H \subseteq B^H$. 

Assume that $i \notin S$. If $x'_i < x_i$, then $i$ is incident to a $C$-edge in $C^H \oplus B^H$, a case we have already excluded. Hence, $x'_i \ge x_i$. Since $G \ne \emptyset$, there is a $j \in S$. We have $x'_i \ge x_i = v_i(A_i) \ge v_j(A_j) \ge x'_j + s$. By Lemma 9.a), $B_i$ is heavy-only. Let $g' \in B_i$, $h \in C_j \cap G$ and $j'$ the owner of $h$ in $B$. We convert $g'$ to a light good and give it to $j'$ and give $h$ back to $j$ as a heavy good. This does not decrease the NSW of $B$, does not change the size of $G$, and yields an allocation at least as close to $C$ as $B$, so we may assume that $i \in S$. 

If $i \in S$, $x'_i < x_i$, and, hence, $C^H \subseteq B^H$. Thus, $g \in B_i$, so we may exchange $g$ with any good in $A_i \cap G$, and the claim holds.
\end{proof}

It is now easy to complete the proof of Theorem~\ref{HalfIntegerResult}. We start with $A^0$. If there is no bundle of value more than $x(A^0) + 1$, $A^0$ is optimal. Otherwise, we select any bundle of maximal value and any good in this bundle, convert the good to a light good and re-optimize to obtain $A^1$. If there is no bundle of value more than $x(A^1) + 1$, we stop. Otherwise, we select a bundle of maximal value and any good in this bundle, convert the good to a light good and re-optimize to obtain $A^2$. We continue in this way and then select the best allocation among $A^0$, $A^1$, \ldots\;.

\paragraph{Polynomial Time:} We construct iteratively allocations $A^1$, $A^2$, \ldots\;. Each time, we convert a heavy good to a light good and re-optimize. There are at most $m$ conversions and each re-optimization takes polynomial time $O(n^4m^{\sfrac{3}{2}})$ by Lemma~\ref{running time}. Thus, the overall time is polynomial.

\section{\classNP-Hardness when $q\ge 3$}\label{sec:NPhard}

In this section, we complement our positive results on polynomial-time NSW optimization. In particular, we show: 

\NPhard*

We provide a reduction from the NP-hard \emph{$q$-Dimensional-Matching ($q$-DM)}.
Given a graph $G$ consisting of $q$ disjoint vertex sets $V_1, \dots, V_q$, each of size $n$, and a set $E\subseteq V_1\times\dots\times V_q$ of $m$, $m \ge n$, edges, decide whether there exists a perfect matching in $G$ or not. Note that for $q=3$ the problem is the well-known 3-DM and, thus, \classNP-hard. \classNP-hardness for $q>3$ follows by simply copying the third set of vertices in the 3-DM instance $q-3$ times, thereby also extending the edges to the new vertex sets.

\paragraph{Transformation:} 
There is one good for each vertex of $G$, call them \emph{vertex goods}. Additionally, there are $p(m-n)$ \emph{dummy goods}. For each edge of $G$, there is one agent who values the $q$ incident vertex goods $\sfrac{p}{q}$ and all other goods $1$. 

\begin{lemma}\label{lemma:NPh1}
    If $G$ has a perfect matching, then there is an allocation $A$ of the goods with $\NSW(A)=p$.  If $G$ has no perfect matching, then for any allocation $A$ of goods,  $\NSW(A) < p$.
\end{lemma}
\begin{proof}
Suppose there exists a perfect matching in $G$. We allocate the goods as follows: Give each agent that corresponds to a matching edge all $q$ incident vertex goods. Now there are $m-n$ agents left. Give each of them $p$ dummy goods. As each agent has utility $p$, the \NSW\ of this allocation is $p$ as well.

For the second claim, assume there is an allocation $A=(A_1,\dots,A_m)$ of goods with $\NSW(A)\geq p$. We show that, in this case, there is a perfect matching in $G$.
First, observe that if we allocate each good to an agent with maximal value for it, we obtain an upper bound on the average utilitarian social welfare of $A$, i.e., $\sfrac{1}{m}\sum_i v_i(A_i) \leq \sfrac{1}{m}(qn\cdot \sfrac{p}{q} + p(m-n)) = p$. 
Applying the AM-GM inequality gives us $\NSW(A)=\left( \prod_i v_i(A_i) \right)^{1/m}\leq p$, and, furthermore, $\NSW(A) = p$ iff $v_i(A_i) = p$ for all agents $i$.
Hence, each agent's utility is $p$ in $A$, and each vertex good is allocated to an incident agent. The next claim allows us to conclude that there are only two types of agents in $A$:

\begin{claim}\label{claimNPh}
    If an agent $i$ has valuation $v_i(A_i) = p$, then she either gets her $q$ incident vertex goods or $p$ other goods.
\end{claim}
\begin{proof} Let $i,j \in \NN_0$ be such that $p = i\cdot \sfrac{p}{q} + j$. Then, $j \le p$, since $i \ge 0$, and $(p -j) q = ip$. Since $p$ and $q$ are co-prime, $p$ divides $j$. Thus, either $j = 0$ and $i = q$, or $j = p$ and $i = 0$. \end{proof}

By Claim~\ref{claimNPh}, an agent either receives $0$ or $q$ of its vertex goods. As there are $qn$ vertex goods, and each of them must be given to an incident agent, there must be $n$ agents receiving their $q$ incident vertex goods, which implies that there is a perfect matching in $G$.
%
%Let $b$ be the number of agents receiving their $q$ incident vertex goods in $\mathcal{A}$. Then $m-b$ is the number of agents receiving $q$ other goods. Since each vertex good must be allocated to an incident agent, $bp=qn\cdot \sfrac{p}{q}$ and thus $b=n$. Hence there must be $n$ agents receiving their $p$ incident vertex goods, which implies that there is a perfect matching in $G$.
\end{proof}

Lemma~\ref{lemma:NPh1} yields the proof of Theorem \ref{NPhard}.

\section{Conclusions}

We have delineated the border between tractable and intractable 2-value NSW instances. If goods have weights $1$ and $\sfrac{p}{q}$ with co-prime $p$ and $q$, instances with $q = 1$ and $q = 2$ are tractable, and for $q \ge 3$ the problem of maximizing the NSW is \classNP-complete. We suggest three  directions for further research.

Our algorithms, in particular for the half-integral case, are fairly complex. Find simpler algorithms.

Find a succinct certificate of optimality in the spirit of McConnell et al.~\cite{McConnell2010}, i.e., can one, in addition, to the optimal allocation compute a succinct and easy-to-check certificate that witnesses the optimality of the allocation. Such certificates are available for many generalized matching problems. See for example the book by Akiyama et al.~\cite{Akiyama-Kano}. An early example is Tutte's certificate for the non-existence of a perfect matching, Tutte~\cite{TutteFactor}. An undirected graph $G$ has no perfect matching if and only if there is a subset $U$ of the vertices such that $\mathit{odd}(G - U) > \abs{U}$, where $\mathit{odd(}G - U)$ is the number of connected components of $G - U$ of odd cardinality. Of course, a perfect matching witnesses the existence of a perfect matching.

In Akrami et al.~\cite{Akrami_Chaudhury_Hoefer_Mehlhorn_Schmalhofer_Shahkarami_Varricchio_Vermande_Wijland_2022} a 1.0345 approximation algorithm is given, and APX-hardness is shown for $q \ge 4$. APX-hardness for $q = 3$ is shown in~\cite{fitzsimmons2024hardness}. It is not known whether the approximation factor is best possible.

\paragraph{Acknowledgement:} We want to thank the anonymous reviewers for their thorough and constructive reviews. It is a pity that we cannot name them. 

% \begingroup \parindent 0pt \parskip 4ex
% \def\enotesize{\normalsize} 
% \theendnotes
% \endgroup

\renewcommand{\htmladdnormallink}[2]{#1}
\bibliographystyle{informs2014}
\bibliography{bibliography}

\end{document}